\documentclass[preprint]{elsarticle}
\usepackage{fullpage,amsmath,amssymb,amsthm,natbib}
\usepackage{epsfig,graphicx,wrapfig,psfrag,color,constants}		
\usepackage{hyperref}
\usepackage{graphicx}
\usepackage{caption}
\usepackage{subcaption}
\graphicspath{{./figs}}
\usepackage{float}
\usepackage{citesort}

\hypersetup{
   colorlinks=true,%
   citecolor=black,%
   filecolor=black,%
   linkcolor=red,%
   urlcolor=blue
}

\def\BD{\Phi}					

\def\bigM{\bigM}
\def\bign{\widetilde{n}}
\def\bigN{\bigN}

\global\long\def\<{\langle}
 \global\long\def\>{\rangle}

 \global\long\def\iunit{\mathrm{i}}

\global\long\def\iunit{\mbox{i}}

\newtheorem{definition}{Definition}
\newtheorem{lemma}{Lemma}
\newtheorem{thm}{Theorem}
\newtheorem{cor}{Corollary}

\newcommand{\tl}[1]{\widetilde{#1}}
\newcommand{\RIPcond}[1]{\delta_{#1}} 
\newcommand{\RIPcondu}[3]{\delta_{#1}(#2,#3)} 
\newcommand{\twonorma}[1]{\left\|#1\right\|_2}
\newcommand{\twonorm}[1]{\|#1\|_2}

\newcommand{\onenorm}[1]{\|#1\|_1}
\newcommand{\zeronorm}[1]{\|#1\|_0}
\newcommand{\cohU}[1]{\mu(#1)} 
\newcommand{\absa}[1]{\left|#1\right|}
\newcommand{\abs}[1]{|#1|}
\newcommand{\XvecU}[2]{X(#1,#2)} 
\newcommand{\QU}[1]{\gamma(#1)} 
\newcommand{\QUsq}[1]{\gamma^2(#1)} 

\newcommand{ \infnorm}[1]{\|#1\|_{\operatorname{max}}}
\newcommand{\svalvec}[1]{\sigma_{#1}} 

\newcommand{\frobnorm}[1]{\|#1\|_F}
\newcommand{\Banachnorm}[1]{\|#1\|_{\Banach}}

\newcommand{\E}[2]{\mathbb{E}^{#2}\hspace{-1mm}\left\{ #1 \right\}}

\newcommand{\Proba}[2]{\operatorname{P}_{#2}\hspace{-1mm}\left\{ #1 \right\}}
\newcommand{\Eb}[2]{\mathbb{E}^{#2}\hspace{0mm}#1} 
\newcommand{\sparses}[1]{\Omega_{#1}} 

\newcommand{\modcohU}[1]{\modcoh(#1)} 
\newcommand{\modcohUsq}[1]{\modcoh^2(#1)} 
\newcommand{\lpnorm}[1]{\ell_{#1}} 
\newcommand{\sgnorm}[1]{\|#1\|_{\psi_2}}
\newcommand{\covera}[3]{\mathcal{N}\left(#1,#2,#3\right)}
\newcommand{\cover}[3]{\mathcal{N}(#1,#2,#3)}
\newcommand{\compress}[1]{[#1]} 

\newcommand{\xnorma}[1]{\left\|#1\right\|_X} 
\newcommand{\xnorm}[1]{\|#1\|_X} 

\newcommand{\ball}[2]{\mathcal{B}^{#1}_{#2}} 

\newcommand{\gennorm}[1]{\|#1\|} 
\newcommand{\sign}[1]{\operatorname{sgn}\left(#1\right)}
\newcommand{\schattena}[2]{\left\|#1\right\|_{S_{#2}}} 
\newcommand{\schatten}[2]{\|#1\|_{S_{#2}}} 
\newcommand{\rank}[1]{\operatorname{Rank}\left(#1\right)}

\newcommand{\pnorm}[2]{\|#1\|_{#2}}
\newcommand{\re}[1]{\mbox{Re}[#1]}
\newcommand{\im}[1]{\mbox{Im}[#1]}
\newcommand{\net}[3]{\mathcal{C}\left(#1,#2,#3\right)} 
\newcommand{\shiftop}[1]{\operatorname{S}^{#1}} 

\def\bigM{\tl{M}}
\def\bigN{\tl{N}}
\def\B{\Phi} 
\def\DBD{\Psi} 
\def\RBD{\Xi} 
\def\basis{U}
\def\basisc{u} 
\def\smallbasis{u} 
\def\rbasis{R} 
\def\rbasisc{r} 
\def\bign{\tl{n}}

\def\Q{\gamma}
\def\Qall{\kappa}	
\def\coh{\mu}
\def\basisc{u} 
\def\rgvec{g} 
\def\dev{t} 
\def\rgmat{G} 
\def\Banach{\chi} 
\def\rad{\xi} 
\def\rv{Z} 
\def\symrv{Y} 
\def\RIPcondgl{\delta} 
\def\thresh{\dummy_0} 
\def\modcoh{\tl{\coh}} 
\def\sgn{\tau} 
\def\maxsgnorm{K} 
\def\genset{\mathcal{S}} 
\def\covres{r} 
\def\rrad{\varepsilon} 
\def\rrads{\epsilon} 
\def\reals{\mathbb{R}}
\def\comps{\mathbb{C}}
\def\supp{T} 
\def\dummy{\nu} 
\def\temp{\beta} 
\def\copies{L} 
\def\tempmat{A} 
\def\tempmatp{B} 
\def\const{C} 
\def\acover{\mathcal{C}} 
\def\dumone{a} 
\def\dumtwo{b} 
\def\dumthree{c} 
\def\circmat{\Gamma} 
\def\circrbd{\Xi} 
\def\extx{\widehat{x}} 
\def\extxtwo{x_e} 
\def\circbasis{T} 
\def\mod{\mbox{mod}} 

\def\parone{E_1}
\def\partwo{E_2}
\def\parthree{E_3}
\def\gammafcn{\gamma}
\def\setofmat{\mathcal{A}}
\def\A{A}

\def\Cover{\mathcal{N}} 

\def\vec{\operatorname{vec}} 

\begin{document}

\title{The Restricted Isometry Property for \\ Random Block Diagonal Matrices}
\date{October 2012}
\author{Armin Eftekhari, Han Lun Yap, Christopher J.\ Rozell, and Michael B.\ Wakin\footnote{The first and second authors contributed equally to this paper. AE and MBW are with the Department of Electrical Engineering and Computer Science at the Colorado School of Mines. HLY and CJR are with the School of Electrical and Computer Engineering at the Georgia Institute of Technology. Email: aeftekha@mines.edu and  yhanlun@dso.org.sg. This work was partially supported by NSF grants CCF-0830456 and CCF-0830320, by NSF CAREER grant CCF-1149225, and by DSO National Laboratories of Singapore. A preliminary version of Theorem~\ref{thm:DBD main thm}, with a different proof, was originally presented at the 2011 IEEE Conference on Information Sciences and Systems (CISS)~\cite{yap2011restricted}.}}

\begin{abstract}
In Compressive Sensing, the Restricted Isometry Property (RIP) ensures that robust recovery of sparse vectors is possible from noisy, undersampled measurements via computationally tractable algorithms. It is by now well-known that Gaussian (or, more generally, sub-Gaussian) random matrices satisfy the RIP under certain conditions on the number of measurements. Their use can be limited in practice, however, due to storage limitations, computational considerations, or the mismatch of such matrices with certain measurement architectures. These issues have recently motivated considerable effort towards studying the RIP for structured random matrices. In this paper, we study the RIP for block diagonal measurement matrices where each block on
the main diagonal is itself a sub-Gaussian random matrix. Our main result states that such matrices can indeed satisfy the RIP but that the requisite number of measurements depends on certain properties of
the basis in which the signals are sparse. In the best case, these matrices perform nearly as well as dense Gaussian random matrices, despite having many fewer nonzero entries.
\end{abstract}
\maketitle


\section{Introduction}

\label{sec:intro}

\label{sec:CS intro}

Many interesting classes of signals have a low-dimensional geometric structure that can be exploited to design efficient signal acquisition and recovery methods. The emerging field of Compressive Sensing (CS) \cite{candes2006compressive,donoho2005neighborliness,donoho2005sparse} deals with signals that can be parsimoniously expressed in a basis or a dictionary. A canonical result in CS states that sparse signals, i.e., signals with very few nonzero entries, can be accurately recovered from a small number of linear measurements by solving a tractable convex optimization problem if the measurement system satisfies the Restricted Isometry Property (RIP)~\cite{candes2006compressive}.

The RIP requires the linear measurement system to approximately maintain the distance between any pair of sparse signals in the measurement space, implying that the geometry of the family of sparse signals is approximately preserved in the measurement space. Apart from playing a central role in the analysis of numerous signal recovery algorithms in CS~\cite{Candes2008589,needell2009cosamp,davenport2010analysis,Fou12SparseRecovery,huang2011recovery}, the RIP also provides a framework to analyze signal processing and inference algorithms in the compressed measurement domain~\cite{davenport2010signal}. Moreover, measurement systems that satisfy the RIP, after undergoing some minor modifications, can approximately preserve the geometry of an arbitrary point cloud (as confirmed by the Johnson-Lindenstrauss lemma)~\cite{krahmer2010new} or a low-dimensional compact manifold~\cite{yap2011stable}.

A measurement system represented by a matrix populated with i.i.d.\ sub-Gaussian\footnote{Roughly speaking, the tail of a sub-Gaussian random variable is similar to that of a Gaussian random variable. This term is defined precisely in Section~\ref{sec:notations}.} random variables is known to satisfy the RIP with high probability whenever the number of rows scales linearly with the sparsity of the signal and logarithmically with the length of the signal~\cite{candes2006compressive,bah2010improved}. Such matrices are also universal in that, with the same number of random measurements, they satisfy the RIP with respect to any fixed sparsity basis with high probability. We refer to such matrices---densely populated with i.i.d.\ random entries---as {\em unstructured} measurement matrices. There has been significant recent interest in studying {\em structured} measurement systems because unstructured random measurements may be undesirable due to memory limitations, computational costs, or specific constraints in the data
acquisition architecture. Many structured systems have been studied in the CS literature, including subsampled bounded orthonormal systems~\cite{rudelson2008fourier,rauhut2010compressive},
random convolution systems (described by partial Toeplitz~\cite{haupt2008toeplitz} and circulant matrices~\cite{rauhut2011restricted,Krahmer2012}) and deterministic matrix constructions~\cite{devore2007deterministic}. Generally, structured random matrices require more measurements to satisfy the RIP and lack the universality of unstructured random matrices.

In this paper, we are concerned with establishing the RIP for block diagonal matrices populated with i.i.d.\ sub-Gaussian random variables. The advantages of such matrices are varied. First, these matrices require less memory and computational resources than their unstructured counterparts. Second, they are particularly useful for representing acquisition systems with architectural constraints that prevent global data aggregation. For example, this type of architecture arises in distributed sensing systems 
 where communication and environmental constraints limit the dependence of each sensor to only a subset of the data and in streaming applications where signals have data rates that necessitate operating on local signal blocks rather than on the entire signal simultaneously. In these scenarios, the data may be divided naturally into discrete blocks, with each block acquired via a local measurement operator.

To make things concrete, for some positive integers $J,N$, and $M$, set $\bigM:=JM$ and $\bigN:=JN$. We model a signal ${x \in \comps^{\bigN}}$ as being partitioned 
into $J$ blocks of length $N$, i.e., $x = \left[ x_1^T, \; \cdots, \; x_J^T \right]^T$ where $x_j \in \comps^{N}$, $j \in \compress{J}$. Here, $\compress{J}$ denotes the set $\{1,2,\cdots,J\}$.
As an example, $x$ can be a video sequence and $\{x_j\}$, $j\in\compress{J}$, can be the individual frames in the video.
For each $j \in\compress{J}$, we suppose that a linear operator $\B_j : \comps^N \rightarrow \comps^M$ collects the measurements $y_j = \B_j x_j$.
In our example, this means that each video frame $x_j$ is measured with an operator $\B_j$. 
Concatenating all of the measurements into a vector $y \in \comps^{\bigM}$, we then have
\begin{equation}
	\underbrace{\left[ \begin{array}{c} y_1 \\ y_2 \\ \vdots \\ y_J \end{array} \right]}_{y: ~ \widetilde{M} \times 1} =
	\left[
	\begin{array}{cccc}
	\BD_1 & & \\ & \BD_2 \\
	& & \ddots & \\
	& & & \BD_J
	\end{array}
	\right]
	\underbrace{\left[ \begin{array}{c} x_1 \\ x_2 \\ \vdots \\ x_J \end{array} \right]}_{x: ~ \widetilde{N} \times 1},
	\label{eq:matrixdef}
\end{equation}
Thus we see that the overall measurement operator relating $y$ to $x$ will have a block diagonal structure. In this paper we consider two scenarios. When the $\{\B_j\}$ are distinct, we call the resulting matrix a Distinct Block Diagonal (DBD) matrix. On the other hand, when the $\{\B_j\}$ are all identical, we call the resulting matrix a Repeated Block Diagonal (RBD) matrix. Our results show that whenever the total number of measurements $\bigM$ is sufficiently large, DBD and RBD matrices can both satisfy the RIP. As we summarize in Sections~\ref{sec:DBD intro} and~\ref{sec:RBD intro} below, the requisite number of measurements depends on the type of matrix (DBD or RBD) and on the basis in which $x$ has a sparse expansion. We also show that certain sparse matrices and random convolution systems considered in the CS literature can be studied in the framework of block diagonal matrices.


In general, proving the RIP for structured measurement systems requires analytic tools beyond the elementary approaches that suffice for unstructured matrices.
For example, in~\cite{rudelson2008fourier} the authors employed tools such as Dudley's inequality from the theory of probabilities in Banach spaces, and in~\cite{rauhut2011restricted} a variant of Dudley's inequality for chaos random processes was used to obtain a result that was out of reach for elementary approaches. While some of the ideas and techniques utilized in~\cite{rudelson2008fourier} and~\cite{rauhut2011restricted} can be used to establish the RIP for random block diagonal matrices (see~\cite{yap2011restricted} for our preliminary study), even these sophisticated tools result in measurement rates that are worse than what we report in this paper. Fortunately, recent work by Krahmer et al. has established an improved bound on the suprema of chaos random processes that enabled them to prove the RIP for Toeplitz matrices with an optimal number of measurements~\cite{Krahmer2012}. The bound in~\cite{Krahmer2012} is very general, and we have leveraged this result for the main results of this paper.
Specifically, the work in~\cite{Krahmer2012} has allowed us to develop a unified treatment of DBD and RBD matrices with bounds on the measurement rates that are significantly improved over our preliminary work.

\subsection{Definition of the RIP}
\label{sec:def and uses of RIP}

A linear measurement operator satisfies the RIP if it acts as an approximate isometry on all sufficiently sparse signals. More specifically, the Restricted Isometry Constant (RIC) of a matrix $\tempmat \in \reals^{\bigM\times\bigN}$ is defined as the smallest positive number $\RIPcond{S}$ for which
\begin{align}\label{eq:ripPlain}
	(1-\RIPcond{S})\twonorm{x}^2\le\twonorm{Ax}^2\le(1+\RIPcond{S})\twonorm{x}^2 & \;\;\;\;\;
	\mbox{for all }x\mbox{ with }\zeronorm{x}\le S,
\end{align}
where $\|\cdot\|_0$ merely counts the number of nonzero entries of a vector. 
In many applications, however, signals may be sparse in an orthobasis $\basis$ other than the canonical basis, and so we will find the notion of the $\basis$-RIP more convenient.
\begin{definition}
	Let $\basis$ denote an orthobasis for $\comps^{\bigN}$. The RIC of a matrix $A\in\reals^{\bigM\times\bigN}$ in the basis $\basis$, $\RIPcond{S}=\RIPcondu{S}{\tempmat}{\basis}$, is defined as the smallest positive number for which
	\begin{align}\label{eq:def of U-RIP}
		(1-\RIPcond{S})\twonorm{x}^2\le\twonorm{Ax}^2\le(1+\RIPcond{S})\twonorm{x}^2 & \;\;\;\;\;
		\emph{\mbox{for all }}x\emph{\mbox{ with }}\zeronorm{\basis^*x}\le S,
	\end{align}
	where $\basis^*$ denotes the conjugate transpose of $\basis$.
\end{definition}
In general, whenever the sparsity basis is clear from the context, $\basis$ is dropped from our notation (after its first appearance).
More generally, the notion of the RIP could be extended to the class of signals that are sparse in an overcomplete dictionary~\cite{candes2011redundant}.
Nonetheless, for the clarity of exposition, we restrict ourselves throughout this paper to considering only orthobases.

\subsection{The RIP for Distinct Block Diagonal (DBD) Matrices}
\label{sec:DBD intro}

Suppose the matrices $\{\B_j\}$, ${j\in\compress{J}}$, in \eqref{eq:matrixdef} are distinct, and let $\DBD$ denote the resulting block diagonal matrix. Following~\cite{park2011concentration}, we say that $\DBD$ has a DBD structure. DBD matrices arise naturally when modeling the information captured by individual (and different) sensors in a sensor network. In this setting, $y_j = \B_j x_j$ represents the local measurements of the signal $x$ made by $j$th sensor and thus $y = \DBD x$ represents the total measurements of $x$ captured by the whole network. Or, as mentioned previously, DBD matrices can be used to represent the process of measuring a video sequence frame-by-frame, but where each frame is observed using a different measurement matrix. DBD matrices also arise in the study of observability matrices in certain linear dynamical systems~\cite{wakin2010observability}, and as another example, DBD
matrices
can be used as a simplified model for the visual pathway (because the
information captured by the photoreceptors in the retina is aggregated locally by horizontal and bipolar cells~\cite{kandel2000principles}). Due to their structure, DBD matrices can be transformed into sparse measurement matrices after permutation of their rows and columns~\cite{SpectralImaging}.
Also see \cite{do2012fast} for an application of block diagonal matrices in developing fast and efficient compressive measurement operators.

We have previously derived concentration inequalities for DBD matrices populated with sub-Gaussian random variables~\cite{wakin2010concentration,park2011concentration}. Rather than ensuring the stable embedding of an entire family of sparse signals, these equalities concern the probability that a bound such as~\eqref{eq:ripPlain} will hold for a single, arbitrary (not necessarily sparse) signal $x$. We have shown that, unlike the case for unstructured random matrices, the probability of concentration with DBD matrices is actually {\em signal dependent}, and in particular the concentration probability depends on the allocation of the signal energy among the signal blocks. However, for signals whose energy is nearly uniformly spread across the $J$ blocks (this happens, for example, with signals that are sparse in the Fourier domain~\cite{park2011concentration}), the highly structured DBD matrices can provide concentration performance that is on par with the unstructured matrices often used in CS.

While concentration of measure inequalities are useful for applications concerning compressive signal processing~\cite{davenport2010signal}, it is not evident how such a concentration result can be extended to give an RIP bound as strong as the one in this paper. Specifically, in Section~\ref{sec:main result DBD} of this paper, we show that if the total number of measurements $\bigM$ scales linearly with the sparsity of the signal $S$ and poly-logarithmically with the ambient dimension $\bigN$, DBD matrices populated with sub-Gaussian random entries will satisfy the $\basis$-RIP with high probability. In addition to a dependence on $S$ and $\bigN$, however, our measurement bounds also reveal a dependence on a property known as the {\em coherence} of the sparsifying basis $\basis$. In this sense, the signal-dependent nature of our concentration of measure inequalities carries over to our RIP analysis for DBD matrices. (The fine details of how this occurs, however, are different.) Our study does confirm that
for the class
of signals that are sparse in the frequency domain, DBD matrices satisfy the RIP with approximately the same number of rows required in an unstructured Gaussian random matrix (despite having many fewer nonzero entries).

\subsection{The RIP for Repeated Block Diagonal (RBD) Matrices}
\label{sec:RBD intro}

Alternatively, suppose the matrices $\{\B_j\}$, ${j\in\compress{J}}$, in \eqref{eq:matrixdef} are all equal, and let $\RBD$ denote the resulting block diagonal matrix. Following~\cite{park2011concentration}, we say that $\RBD$ has an RBD structure. In the context of the sensor network, video processing, and observability applications discussed before, RBD matrices arise when the same measurement matrix is used for all the signal blocks. In the delay embedding of dynamical systems, as another example, a time series is obtained by repeatedly applying a scalar measurement function to the trajectory of a dynamical system. This time series can then be embedded in a low-dimensional space (hence the name), and this embedding can be expressed through an RBD measurement matrix provided that the scalar measurement function is linear~\cite{5930379}. Though not obvious at first glance, RBD matrices also have structural similarities with random convolution
matrices found in the CS literature~\cite{rauhut2011restricted,Krahmer2012}. We revisit this connection in order to re-derive the RIP for partially circulant random matrices in Section~\ref{sec:Toeplitz}.

We have previously derived concentration inequalities for RBD matrices populated with Gaussian random variables~\cite{wakin2010concentration,rozell2010concentration,park2011concentration}. Our bounds for these matrices again reveal that the probability of concentration is signal dependent. However, in this case, our bounds depend on both the allocation of the signal energy among the signal blocks as well as the mutual orthogonality of the signal blocks.

In Section~\ref{sec:main result RBD} of this paper, we show that if the total number of measurements $\bigM$ scales linearly with $S$ and poly-logarithmically with $\bigN$, RBD matrices populated with sub-Gaussian random variables will satisfy the $\basis$-RIP with high probability. Our measurement bounds also reveal a dependence on a property known as the {\em block-coherence} of the sparsifying basis $\basis$ that quantifies the dependence between its row blocks. When the block-coherence of $\basis$ is small, RBD matrices perform favorably compared to unstructured Gaussian random matrices. Most sparsifying bases are in fact favorable in this regard; we prove that the block-coherence of $\basis$ is small when $\basis$ is selected randomly. Once again, for RBD matrices, the signal dependent nature of concentration inequalities and the dependence of the RIP on the sparsifying basis emerge as two sides of the same coin.

\subsection{Outline}

This paper is organized as follows. Section~\ref{sec:notations} introduces the notation used throughout the rest of the paper. Section~\ref{sec:main results} summarizes our main results regarding the RIP for DBD and RBD matrices; these results are later proved in Section~\ref{sec:proof of DBD RBD main result}. Section~\ref{sec:simulations} presents numerical simulations that illustrate the dependence of signal recovery performance on the sparsifying basis $U$. We conclude the paper with a short discussion in Section~\ref{sec:conclusion}. We note that, for the reader's convenience, the Toolbox (\ref{sec:toolbox}) gathers some general tools from linear algebra and probability theory used in our analysis.

\section{Notation}

\label{sec:notations}

We reserve the letters $C,C_1,C_2,\cdots$ to represent universal positive constants. We adopt the following (semi-)order: $\dumone\lesssim\dumtwo$ means that there is an absolute constant $\Cl{notation1}$ such that $\dumone\le\Cr{notation1}\dumtwo$. If, instead of being an absolute constant, $\Cr{notation1}$  depends on some parameter $\dumthree$, we write $\dumone\lesssim_{\dumthree}\dumtwo$. Also $\dumone\gtrsim\dumtwo$ and $\dumone\gtrsim_{\dumthree}\dumtwo$ are defined similarly.

For an integer $S$, a signal with no more than $S$ nonzero entries is called $S$-sparse, and $S$ is known as the sparsity level. In particular, $\zeronorm{\dumone}$ denotes the number of nonzero entries of a vector $\dumone$. More generally, a signal that is a linear combination of at most $S$ columns of a basis is said to be $S$-sparse in that basis. The conjugate transpose of a matrix $\tempmat$ will be denoted by $\tempmat^*$. In this paper, $\rank{\tempmat}$ stands for the rank of matrix $\tempmat$. In addition to the regular $\lpnorm{p}$-norms in the Euclidean spaces, $1\le p\le\infty$, we use $\twonorm{\tempmat}$ and $\frobnorm{\tempmat}$ to denote the spectral and Frobenious norms of a matrix $\tempmat$, respectively. We use $ \infnorm{\tempmat}$ to denote the largest entry of the matrix $\tempmat$ in magnitude. For $1\le p\le\infty$, the Schatten norm of order $p$ of a matrix $\tempmat$ is
denoted by $\schatten{\tempmat}{p}$ and is defined as
\[
	\schatten{\tempmat}{p}:=\pnorm{\svalvec{\tempmat}}{p},
\]
where $\svalvec{\tempmat}$ is the vector formed by the singular values of $\tempmat$.  Observe that $\schatten{\tempmat}{\infty} = \|\tempmat\|_2$ and $\schatten{\tempmat}{2} = \|\tempmat\|_F$. Throughout this paper, for a matrix $\tempmat$, $\vec(\tempmat)$ returns the vector formed by stacking the columns of $\tempmat$. Also, we will use the conventions $\compress{N}:=\{1,2,\cdots,N\}$ (for an integer $N$) 
and $\#T$ for the cardinality of a set $T$.

When it appears, the subscript of an expectation operator $\Eb{}{}$ specifies the (group of) random variable(s) with respect to which the expectation is taken. 
For a random variable $\rv$ taking values in $\comps$, we define $\Eb{\abs{\rv}}{p}:=(\Eb{\abs{\rv}^p}{})^{1/p}$, $p\ge 1$. A random variable $\rv$ is sub-Gaussian if its sub-Gaussian norm, defined below, is finite~\cite{vershynin2010introduction}:
\begin{equation}\label{eq:def of sub-gaussian norm}
	\sgnorm{\rv}:=\sup_{p\ge1}\frac{1}{\sqrt{p}}\Eb{\abs{\rv}}{p}.
\end{equation}
Qualitatively speaking, the tail of (the distribution of) a sub-Gaussian random variable is similar to that of a Gaussian random variable, hence the name.
A Rademacher sequence is a sequence of i.i.d.\ random variables that take the values $\pm1$ with equal probability.  In this paper, every appearance of such a sequence is independent of other random variables in this paper and, in particular, independent of other Rademacher sequences elsewhere in the paper. Finally, $\underset{\tiny{\mbox{i.d.}}}{=}$ means that the random variables on both sides of the equality have the same distribution.

A set $\net{\genset}{\gennorm{\cdot}}{\covres}$ is called a cover for the set $\genset$ at resolution $\covres$ and with respect to the metric $\gennorm{\cdot}$ if for every $x\in\genset$, there exists $x'\in\net{\genset}{\gennorm{\cdot}}{\covres}$ such that $\gennorm{x-x'}\le\covres$. The minimum cardinality of all such covers is called the covering number of $\genset$ at resolution $\covres$ and
with respect to the norm $\|\cdot\|$, and is denoted here by $\cover{\genset}{\|\cdot\|}{\covres}$. 

\section{Main Results}\label{sec:main results}




\subsection{Measures of Coherence}
\label{sec:sparsifying basis}

Our results for random block diagonal matrices depend on certain properties of the sparsity basis, i.e., the basis in which the signals have a sparse expansion. These properties are defined and studied in this section;  this sets the stage for a detailed statement of our main results in Sections~\ref{sec:main result DBD} and~\ref{sec:main result RBD}.

\subsubsection{Coherence Definitions}

The {\em coherence} of an orthobasis $\basis\in\comps^{\bigN\times\bigN}$ is defined as follows~\cite{candes2008introduction}:
\begin{equation}\label{eq:def of coherence}
	\cohU{\basis}:=\sqrt{\bigN}\max_{p,q\in\compress{\bigN}}\absa{\basis(p,q)},\end{equation}
where $\basis(p,q)$ is the $(p,q)$th entry of $\basis$. If $\{\basisc_{\bign}\}$ and $\{e_{\bign}\}$, $\bign \in \compress{\bigN}$, denote the columns of $\basis$ and of the canonical basis for $\comps^{\bigN}$, respectively, one can easily verify that
\begin{equation}\label{eq:eq. def of coherence}
	\cohU{\basis}=\sqrt{\bigN}\max_{p,q\in\compress{\bigN}}\absa{\<\basisc_{p},e_{q}\>}.
\end{equation}
This allows us to interpret $\cohU{\basis}$ as the similarity between $\basis$ and the canonical basis.

A few more definitions are in order before we can define the second important property of a basis used in this paper. For $\alpha\in\comps^{\bigN}$, set $x(\alpha)=x(\alpha,\basis):=\basis\alpha$, and define $x_j(\alpha)=x_j(\alpha,\basis)\in\comps^N$, $j\in\compress{J}$, such that
\begin{equation}\label{eq:breaking x into chunks}
	x(\alpha)=[x_1(\alpha)^T,x_2(\alpha)^T,\cdots,x_J(\alpha)^T]^T.
\end{equation}
%
If we also define $\basis_j\in\comps^{N\times\bigN}$, ${j\in\compress{J}}$, such that \begin{equation}\label{eq:def of U1 UJ}
	\basis=[\basis_1^T,\basis_2^T,\cdots,\basis_J^T]^T,
\end{equation}
we observe that $x_j(\alpha) = U_j \alpha$ for every $j$. Define $X_R(\alpha,\basis)\in\comps^{N\times J}$ as
\[
	X_R(\alpha)=X_R(\alpha,\basis):=\left[\begin{array}{cccc}
	x_{1}(\alpha) & x_{2}(\alpha) & \cdots & x_{J}(\alpha)\end{array}\right]= \left[ \begin{array}{ccc} U_1 \alpha & \cdots & U_J \alpha \end{array} \right].
\]
Now the {\em block-coherence} of $\basis$, denoted by $\QU{\basis}$, is defined as
\begin{equation}\label{eq:def of block-coherence}
	\QU{\basis}:=\sqrt{J}\max_{\bign\in\compress{\bigN}}\twonorma{X_R(e_{\bign},\basis)}.
\end{equation}
In words, $\QU{\basis}$ is proportional to the maximal spectral norm when any column of $U$ is reshaped into an $N \times J$ matrix. In the special case where $J=1$ or $N=1$, $\gamma(U)$ simply scales with the maximum $\ell_2$-norm of the columns of $U$. In analogy with \eqref{eq:eq. def of coherence}, one can also think of \eqref{eq:def of block-coherence} as a (non-commutative)  coherence measure between $\basis$ and $I_{\bigN}$.\footnote{It can be easily verified that, in general, $\max_{\bign}\|X_R(e_{\bign},\basis)\|_2\ne \max_{\bign}\|X_R(\smallbasis_{\bign},I_{\bigN})\|_2$, where $\smallbasis_{\bign}$ is the $\bign$th column of $\basis$.}
Qualitatively speaking, $\Q(\basis)$ measures the orthogonality and distribution of energy between $U_1,U_2,\cdots,U_J$, the consecutive $N\times\widetilde{N}$ row-submatrices of $\basis$. If for every column of $\basis$, the energy is evenly distributed between its row-blocks and they are nearly orthogonal, $\Q(\basis)$ will be small and, as we will see later, better suited for our purposes. In the next subsection, we compute the coherence and block-coherence of a few widely-used orthobases.

\subsubsection{Computing the Coherence for a Few Orthonormal Bases}

It is easily verified that
\begin{equation}\label{eq:range of coherence}
	1\le\cohU{\basis}\le\sqrt{\bigN}.
\end{equation}
The upper bound is achieved, for example, by the canonical basis in $\comps^{\bigN}$, i.e.,  $\cohU{I_{\bigN}}=\sqrt{\bigN}$. The lower bound, on the other hand, is achieved by any basis that is maximally incoherent with the canonical basis. For example, $\cohU{F_{\bigN}}=1$, where $F_{\bigN}$ denotes the Fourier basis in $\comps^{\bigN}$. The next lemma, proved in \ref{sec:proof of coh of rand basis}, indicates that {\em most} orthobases are also highly incoherent with the canonical basis. 
\begin{lemma}\label{lem:coh of rand basis}
	Let $\rbasis\in\reals^{\bigN\times\bigN}$ denote a generic orthobasis in $\comps^{\bigN}$ chosen randomly from the uniform distribution on the orthogonal group.
	Then the following holds for fixed $\dev\gtrsim 1$ and $\bigN\gtrsim\dev^2\log\bigN$:
	\begin{equation}\label{eq:coh of rand basis}
		\Proba{\cohU{\rbasis}>\dev\sqrt{\log\bigN}}{}\lesssim\bigN^{-\dev}.
	\end{equation}

\end{lemma}

We now turn to computing the block-coherence  of the same orthobases. Since every column of $U$ has unit $\lpnorm{2}$-norm, it is easily observed that
%
%
\begin{equation}
	1\le \QU{\basis}\le\sqrt{J}.\label{eq:range of Q}
\end{equation}
Consider the canonical basis in $\comps^{\bigN}$.
For every $\bign\in\compress{\bigN}$, $\XvecU{e_{\bign}}{I_{\bigN}}$ has a single non-zero entry, which equals $1$, and thus $\twonorm{\XvecU{e_{\bign}}{I_{\bigN}}}=1$. Hence,
$\QU{I_{\bigN}}=\sqrt{J}$.
Moving on to $F_{\bigN}$, we observe that the entries of the first column of $F_{\bigN}$ equal $\bigN^{-1/2}$. As a result, the entries of $\XvecU{e_1}{F_{\bigN}}$ all equal $\bigN^{-1/2}$. It follows that $\twonorm{\XvecU{e_1}{F_{\bigN}}}=1$ and therefore $\QU{F_{\bigN}}= \sqrt{J}$.

Because the canonical basis and the Fourier basis---which one might naturally consider to be opposite ends on some spectrum of orthobases---both have large block-coherence, one might wonder whether any orthobasis could have small block-coherence. As we will see, {\em most} possible orthobases actually do have small block-coherence. For example, consider the generic orthobasis $\rbasis$ constructed in Lemma~\ref{lem:coh of rand basis}. The columns of $\XvecU{e_1}{\rbasis}$ are $J$ random vectors in $\reals^{N}$. These vectors are weakly dependent because the first column of $\rbasis$ has unit $\ell_2$-norm. 
With high probability, the length of each vector is  approximately $1/\sqrt{J}$ (so that the $\ell_2$-norm of the first column of $\rbasis$ is one). If $J\le N$, then with high probability these points are spread out in $\reals^{N}$ so that $\twonorm{\XvecU{e_{1}}{\rbasis}} \approx1/\sqrt{J}$. Now, since the columns of $\rbasis$ have the same distribution, $\twonorm{\XvecU{e_{\bign}}{\rbasis}} \approx1/\sqrt{J}$ for every $\bign\in\compress{\bigN}$.
 Therefore, $\QU{\rbasis}\approx1$, which is much smaller than the block-coherence of the
canonical and Fourier bases. The next result, proved in \ref{sec:Proof-of-Lemma Q for random mtx}, formalizes this discussion.

\begin{lemma}\label{lem:Q of random matrix}
	Consider the generic orthobasis $\rbasis$ constructed in Lemma~\ref{lem:coh of rand basis}. For fixed $\dev\le1$, the following holds if $J\le N$ and $N\gtrsim\dev^{-2}\log \bigN$:
	\begin{equation}\label{eq:bounds on Q of rand matrix}
		\Proba{\QU{\rbasis}\gtrsim 1+\sqrt{\frac{J}{N}}+\dev}{}\lesssim\bigN^{-\dev}.
	\end{equation}
\end{lemma}

We close this section by noting that Section~\ref{sec:Toeplitz} provides an example of a deterministic basis with small block-coherence, which is obtained by  modifying the Kronecker product of the canonical and Fourier bases. (This basis is then used to prove the RIP for partial random circulant matrices.)

\subsection{The RIP for DBD Matrices}\label{sec:main result DBD}

Let $\B\in\reals^{M\times N}$ denote a matrix populated with i.i.d.\ sub-Gaussian random variables having mean zero, standard deviation $1/\sqrt{M}$, and sub-Gaussian norm $\sgn/\sqrt{M}$, for some $\sgn>0$.
Take $\{\B_j\}$ in \eqref{eq:matrixdef} to be $J$ independent copies of $\Phi$ and let $\DBD\in\reals^{\bigM\times\bigN}$ denote the resulting block diagonal matrix in \eqref{eq:matrixdef}.
Our first main result, proved in Section~\ref{sec:proof of DBD RBD main result}, establishes the RIP for DBD matrices with this construction.
\begin{thm}\label{thm:DBD main thm}
	Let $\basis$ denote an orthobasis for $\comps^{\bigN}$ and  set
\begin{equation}\label{eq:def of mutilde}
	\modcohU{\basis}:=\min\left(
	\sqrt{J},\; \cohU{U}\right),\qquad\mbox{where}\qquad\mu(U)=\sqrt{\widetilde{N}}\max_{p,q\in[\widetilde{N}]}|U(p,q)|.
\end{equation}
	If $S \gtrsim 1$ and
	\begin{equation}\label{eq:req bigM for DBD RIP}
		\bigM \gtrsim_{\sgn} \RIPcondgl^{-2}\modcohUsq{\basis} \cdot S\cdot\log^2 S\log^{2}\bigN,
	\end{equation}
	then $\RIPcondu{S}{\DBD}{\basis}\le\RIPcondgl<1$, except with a probability of at most $O(\bigN^{-\log\bigN\log^2 S})$. 
\end{thm}
A few remarks are in order. The requisite number of measurements is linear in the sparsity level $S$ and (poly-)logarithmic in the ambient dimension $\bigN$, on par with an unstructured random Gaussian matrix~\cite{candes2006compressive}. More importantly, the requisite number of measurements scales with $\modcohUsq{\basis}$ which takes a value in the interval $[1,J]$. For the Fourier basis, we calculated that $\cohU{F_{\bigN}} = 1$. Therefore, when measuring signals that are sparse in the frequency domain, we observe that a DBD matrix compares favorably to an unstructured Gaussian matrix of the same size. This is in the sense that they both require the same number of measurements to achieve the RIP (up to a  poly-logarithmic factor).

On the other hand, when the orthobasis $\basis$ is highly coherent with the canonical basis, the requisite number of measurements is proportional to $SJ$ (instead of $S$). While possibly unfavorable, this is indeed necessary (to within a poly-logarithmic factor) to achieve the RIP in some cases. For example, recall that $\cohU{I_{\bigN}}=\sqrt{\bigN}$ and so $\modcohU{I_{\bigN}}=\sqrt{J}$.  To see why the results are optimal in this case, consider the class of $S$-sparse signals in $I_{\bigN}$ whose nonzero entries are located within the first length-$N$ block of the signal. Achieving a stable embedding of this class of signals requires $\B_1$ itself to satisfy the RIP. This matrix, $\B_1$, is an unstructured sub-Gaussian matrix, and ensuring that it satisfies the RIP requires  $M\gtrsim_{\sgn}\RIPcondgl^{-2} S \log(N/S)$~\cite{candes2006compressive}. Consequently, ensuring the $I_{\bigN}$-RIP for $\DBD$ is only possible when $\bigM\gtrsim_{\sgn}\RIPcondgl^{-2} JS \log(N/S)$, as predicted by Theorem~\ref{thm:DBD main thm} (up to a poly-logarithmic term). The required number of measurements in this case can still be parsimonious ($\bigM \ll \bigN$), however, if the sparsity level $S$ of the signal
$x$ is much less than $N$, the length of each signal block $x_j$.

As a final note, Theorem~\ref{thm:DBD main thm} implies the RIP for a certain class of sparse matrices which are of potential interest in their own right~\cite{berinde2008combining,SpectralImaging}.
\begin{cor}
	Let $\basis$ denote an orthobasis for $\comps^{\bigN}$ and define $\modcohU{\basis}:=\min\left(
	\sqrt{J},\; \cohU{U}\right)$.
	Let $\DBD'$ denote the (sparse) matrix obtained by an arbitrary permutation of the rows and columns of $\DBD$. 	If $S \gtrsim 1$ and
	\begin{equation}\label{eq:req bigM for DBD RIP}
		\bigM \gtrsim_{\sgn} \RIPcondgl^{-2}\modcohUsq{\basis} \cdot S\cdot\log^2 S\log^{2}\bigN,
	\end{equation}
	then $\RIPcondu{S}{\DBD'}{\basis}\le\RIPcondgl<1$, except with a probability of at most $O(\bigN^{-\log\bigN\log^2 S})$.
\end{cor}
\begin{proof}
	Without loss of generality, consider no permutation in the rows and let $P_c\in\reals^{\bigN}$ denote the permutation matrix for the columns of $\DBD$. Since $P_c\basis$ has the same coherence as $\basis$, the claim follows by applying Theorem~\ref{thm:DBD main thm}.
\end{proof}

\subsection{The RIP for RBD Matrices}\label{sec:main result RBD}

\subsubsection{Main Result for RBD Matrices}

Let $\B\in\reals^{M\times N}$ denote a matrix populated with i.i.d.\ sub-Gaussian random variables having mean zero, standard deviation $1/\sqrt{M}$, and sub-Gaussian norm $\sgn/\sqrt{M}$, for some $\sgn>0$. Take $\B_j=\B$ for every $j\in\compress{J}$ in \eqref{eq:matrixdef} and let $\RBD\in\reals^{\bigM\times\bigN}$ denote the resulting block diagonal matrix in \eqref{eq:matrixdef}. Our second main result, also proved in Section~\ref{sec:proof of DBD RBD main result}, establishes the RIP for RBD matrices with this construction.
\begin{thm}\label{thm:RBD main thm}
	Let $\basis$ denote an orthobasis for $\comps^{\bigN}$.
	If $S \gtrsim 1$ and
	\begin{equation*}
		\bigM \gtrsim_\sgn  \RIPcondgl^{-2}\QUsq{\basis}\cdot S \cdot \log^2 S\log^{2}\bigN,
	\end{equation*}
	then $\RIPcondu{S}{\RBD}{\basis}\le\RIPcondgl<1$, except with a probability of at most $O(\bigN^{-\log\bigN\log^2 S})$. 
\end{thm}

A few remarks are in order.
At one end of the spectrum, the block-coherence of an orthobasis could equal $\sqrt{J}$  and consequently the required number of measurements above would scale with $JS$. This happens for signals that are sparse, for example, in the time (canonical basis) or frequency domains.
Our result is indeed optimal for both of these bases (up to a poly-logarithmic factor).
The same argument for the canonical basis carries over from the DBD matrices. For the Fourier basis, we note that it is possible to construct certain classes of periodic signals in $\comps^{\bigN}$ that would require the lower bound on $\bigM$ to scale with $JS$.
Consider, for example, the class of signals consisting of all $S$-sparse combinations of columns $1, J+1, \dots, (J-1)N+1$ from $F_{\bigN}$. If $x$ belongs to this class, then, by construction, $x_1,x_2,\cdots,x_J$ (as defined in \eqref{eq:breaking x into chunks}) are all equal because $x$ is periodic with a period of $N$.
As a result, different blocks of $\RBD$ take the same measurements from $x$. Therefore, as was the case with the DBD matrices, obtaining a stable embedding of this class of signals requires $M\gtrsim_{\sgn}\RIPcondgl^{-2} S \log(N/S)$, and equivalently, $\bigM\gtrsim_{\sgn}\RIPcondgl^{-2} JS \log(N/S)$.

At the other end of the spectrum, for a generic orthobasis $\rbasis$ we computed that $\QU{\rbasis}\lesssim 1$ with high probability.
For signals that are sparse in this basis (and therefore for many possible orthobases in general), an RBD matrix performs nearly as well as an unstructured Gaussian random matrix. The RBD structure allows us to prove the RIP for certain classes of structured random matrices as a special case of Theorem~\ref{thm:RBD main thm}. In particular, in the next subsection we re-derive an RIP bound for partial random circulant matrices that originally appeared in~\cite{Krahmer2012}. As a byproduct, we also construct a deterministic sparsity basis that achieves a performance similar to the generic orthobasis we have considered above.

\subsubsection{The RIP for Partial Random Circulant Matrices}
\label{sec:Toeplitz}

This section demonstrates that the RBD model, together with Theorem~\ref{thm:RBD main thm}, can be used to derive the RIP for partial random circulant matrices.\footnote{The arguments in this section extend without much effort to the more general case of establishing the RIP for partial random Toeplitz matrices.
}
More specifically, we focus on proving the RIP for $\circmat\in\reals^{J\times P}$, with $J\le P$, defined as
\begin{eqnarray*}
	\circmat = \frac{1}{\sqrt{J}}\left[
	\begin{array}{cccc}
		\rrads_1 & \rrads_2 & \cdots & \rrads_P \\
		\rrads_P & \rrads_1 & \cdots & \rrads_{P-1} \\
		\vdots & \vdots &  & \vdots \\
		\rrads_{P-J+2} & \rrads_	{P-J+3} & \cdots & \rrads_{P-J+1}
	\end{array}
	\right],
\end{eqnarray*}
where $\{\rrads_p\}$, $p\in\compress{P}$, is a sequence of i.i.d.\ zero-mean, unit-variance random variables with sub-Gaussian norm $\sgn$. We let $\rrad$ denote the vector formed by $\{\rrads_p\}$.
%
In order to use Theorem~\ref{thm:RBD main thm} in this setting, we make the following argument: for any signal $x \in \comps^P$, we can write $\circmat x$ as the multiplication of an RBD matrix $\circrbd \in \reals^{J \times PJ}$ and an extended vector $\extx \in \comps^{PJ}$:
\begin{eqnarray}
	\circmat x =
	\overset{\Huge{\circrbd}}{
	\overbrace{
	\left[
	\begin{array}{cccc}
		\rrad^* & & & \\
		& \rrad^* & & \\
		& & \ddots & \\
		& & & \rrad^*
	\end{array}
	\right]}
	}	
	\cdot \frac{1}{\sqrt{J}}
	\overset{\extx}{
	\overbrace{
	\left[
	\begin{array}{c}
		\shiftop{0} x \\
		\shiftop{1} x \\
		\vdots \\
		\shiftop{J-1} x
	\end{array}
	\right]
	}
	}
	=:\circrbd\cdot\frac{1}{\sqrt{J}}\extx,
	\label{eq:PC_RBD_formulation}
\end{eqnarray}
where $\shiftop{}$ is the cyclic shift-up operator on column vectors in $\comps^P$. The next lemma (proved in \ref{sec:orthobasis for Toeplitz}) states that if $x$ is sparse, then $\extx$ has a sparse representation in a favorable orthobasis $\circbasis$ (constructed in the proof).
\begin{lemma}
	\label{lem:orthobasis for Toeplitz}
	There exists an orthobasis $\circbasis$ with $\QU{\circbasis}=1$, such that every $\extx$ has an $S$-sparse representation in $\circbasis$ if the corresponding $x$ is $S$-sparse. That is, for every such $\extx$, there exists $\extxtwo$ with $\|\extxtwo\|_0\le S$ such that $\extx/\sqrt{J}=\circbasis\extxtwo$.
\end{lemma}
Therefore, the $\circbasis$-RIP for $\circrbd$ implies the RIP for $\circmat$.
To be more specific, after setting $M=1$, Theorem~\ref{thm:RBD main thm} implies that $\RIPcond{S}(\circmat)\le\RIPcond{S}(\circrbd,T)\le\RIPcondgl$ except with a probability of at most $O(P^{-\log P\log^2 S})$, provided that
\[
	J\gtrsim_{\sgn} \RIPcondgl^{-2}\cdot S \log^2 S\log^{2}P,
\]
which is equivalent to Theorem 1.1 in~\cite{Krahmer2012}.



\section{Numerical Simulations}
\label{sec:simulations}

This section contains a series of simulations that are intended to reinforce our findings for the reader. An ideal scenario might involve generating several random block diagonal matrices while varying the sparsity level $S$ and the number of measurements $\bigM$ and measuring the fraction of realizations in which the RIC falls below a fixed threshold. Unfortunately, checking the RIP for a matrix is known to be an NP~hard problem~\cite{tillmann2012computational}, although certain algorithms~\cite{journee2010generalized,dossal2010numerical} have been proposed for computing lower bounds on the RIC for moderately sized random matrices. In particular, \cite{dossal2010numerical} contains a fast heuristic algorithm for this purpose.

As we have discussed in Section~\ref{sec:CS intro}, one context where the RIP often arises is in the analysis of sparse signal recovery algorithms such as Basis Pursuit (BP)~\cite{Candes2008589}. In fact, for dense random matrices, it is known that the RIP serves as a rather pessimistic criterion to ensure the success of sparse signal recovery algorithms (particularly BP); we refer the interested reader to~\cite{blanchard2011compressed,blanchard2011phase,donoho2009observed,monajemi2013deterministic} for a more comprehensive discussion. In this section, we use empirical signal recovery performance---with a limited number of signals---to illustrate how various choices for the sparsity basis will interact with a block diagonal measurement matrix. To be clear, this is different from checking whether the RIP holds. What we see, however, is that the same bases which are most favorable for our RIP bounds tend to also be most favorable for empirical signal recovery performance.

We now detail the setup in this section. With $N=100$ and $J=10$ (and consequently, $\bigN=1000$), we generate a single $\bigN\times\bigN$ DBD matrix whose entries are i.i.d.\ Gaussian random variables having zero mean and unit standard deviation. Note that each diagonal block is of size $N\times N$. For each pair $(S,M)\in\compress{\bigM}\times \compress{N}$, the following procedure is executed. An $\bigM\times\bigN$ DBD matrix $\DBD$ is formed by keeping (and appropriately normalizing) the first $M$ rows of each diagonal block of the large $\bigN\times\bigN$ matrix. Then $20$ random $S$-sparse signals in the canonical basis of $\comps^{\bigN}$ are generated. These sparse signals are measured using $\DBD$ and reconstructed back (from incomplete measurements) via BP.\footnote{In order to implement BP, we used YALL1, a package for solving $\ell_1$ problems~\cite{YALL1,yang2011alternating}.
} We deem the recovery successful if the relative $\ell_2$ reconstruction error is less than the fixed threshold $10^{-2}$.  The fraction of successful recovery is recorded and this procedure is repeated for other pairs  $(S,M)$. To be clear, we use a fixed measurement matrix to generate the data corresponding to each pair $(S,M)$. As mentioned earlier, the $\widetilde{M}\times\widetilde{N}$ measurement matrix itself is generated from the $\widetilde{N}\times\widetilde{N}$ matrix fixed throughout the experiment. The resulting phase transition graph is depicted in Figure~\ref{fig:DBD simulation canonical}, where the color of each pixel ranges from black for perfect recovery in every realization to white for failed recovery every time.

We repeat the above simulation for signals that are sparse in the Fourier and generic bases (see Lemma~\ref{lem:coh of rand basis}). A new generic basis is generated in each iteration. All of the above simulations are then repeated with an RBD matrix. The results of the simulations are displayed in Figures~\ref{fig:DBD simulations} and~\ref{fig:RBD simulations}. In the simulations with DBD matrices, recovery of canonical and frequency sparse signals are, respectively, the least and most successful of the three instances. Signals that are sparse in the generic bases can be recovered nearly as well as frequency sparse signals.

In the simulations with RBD matrices, signals that are sparse in the generic bases are recovered best, while the results for the canonical and Fourier bases are less satisfactory. These observations mirror our findings in Section~\ref{sec:main results}. We do point out, however, that for the case of RBD measurement matrices we are able to recover frequency sparse signals somewhat better than signals that are sparse in the canonical basis. We note that such difference is not reflected in our RIP bounds. While this performance difference could be due to a lack of explicit numerical constants in our results, it is more likely an artifact of our simulations: We are not directly confirming the RIP but rather testing the recovery of randomly generated test signals, and those few signals which make it difficult to satisfy the RIP may be more pathological in the Fourier basis than in the canonical basis.

\begin{figure}[H]
	\centering

	\begin{subfigure}[b]{0.45\linewidth}
	  \centering
	  \includegraphics[scale=.40]{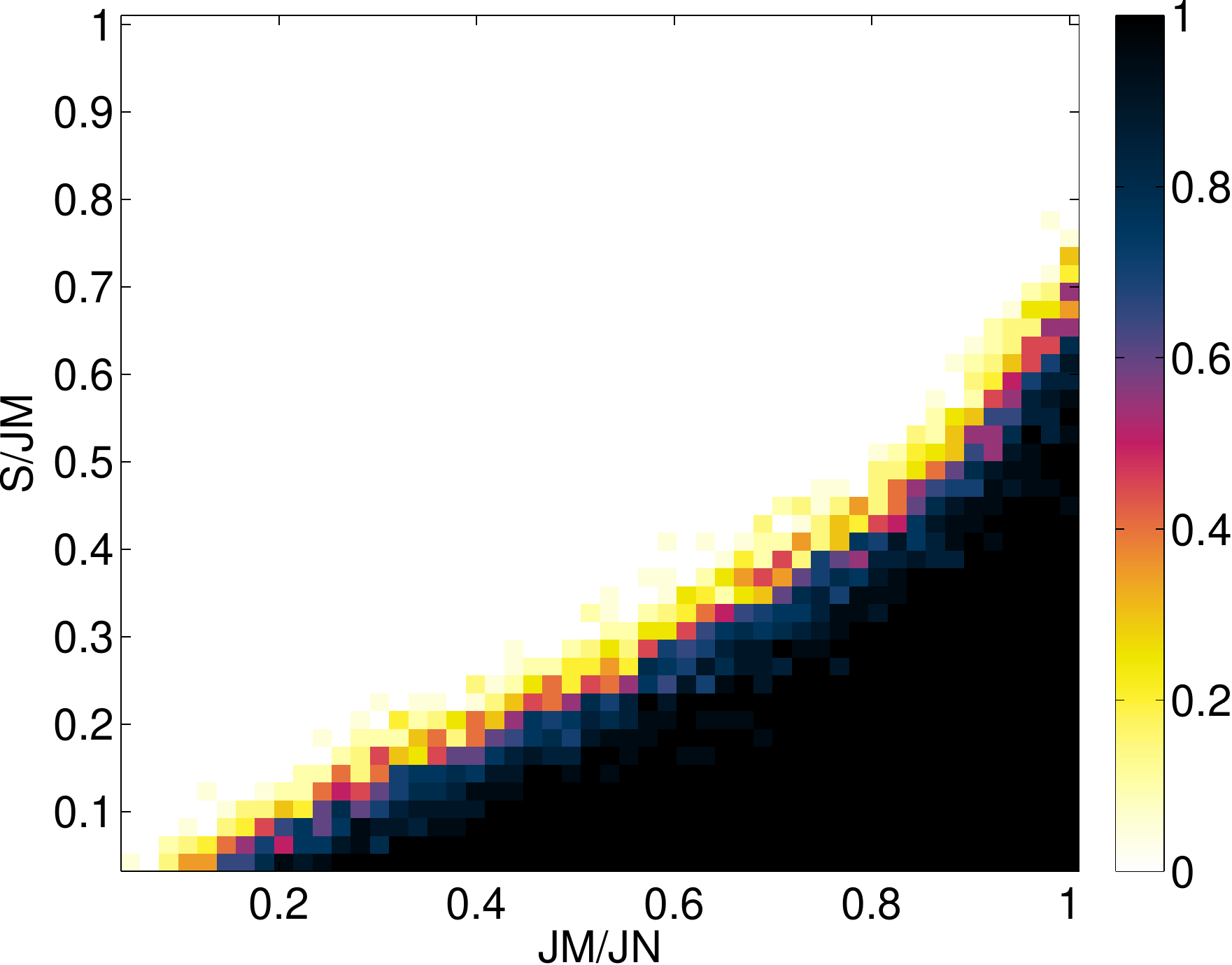}
	  \caption{Canonical basis}
	  \label{fig:DBD simulation canonical}
	\end{subfigure}
	\begin{subfigure}[b]{0.45\linewidth}
	  \centering
	  \includegraphics[scale=.4]{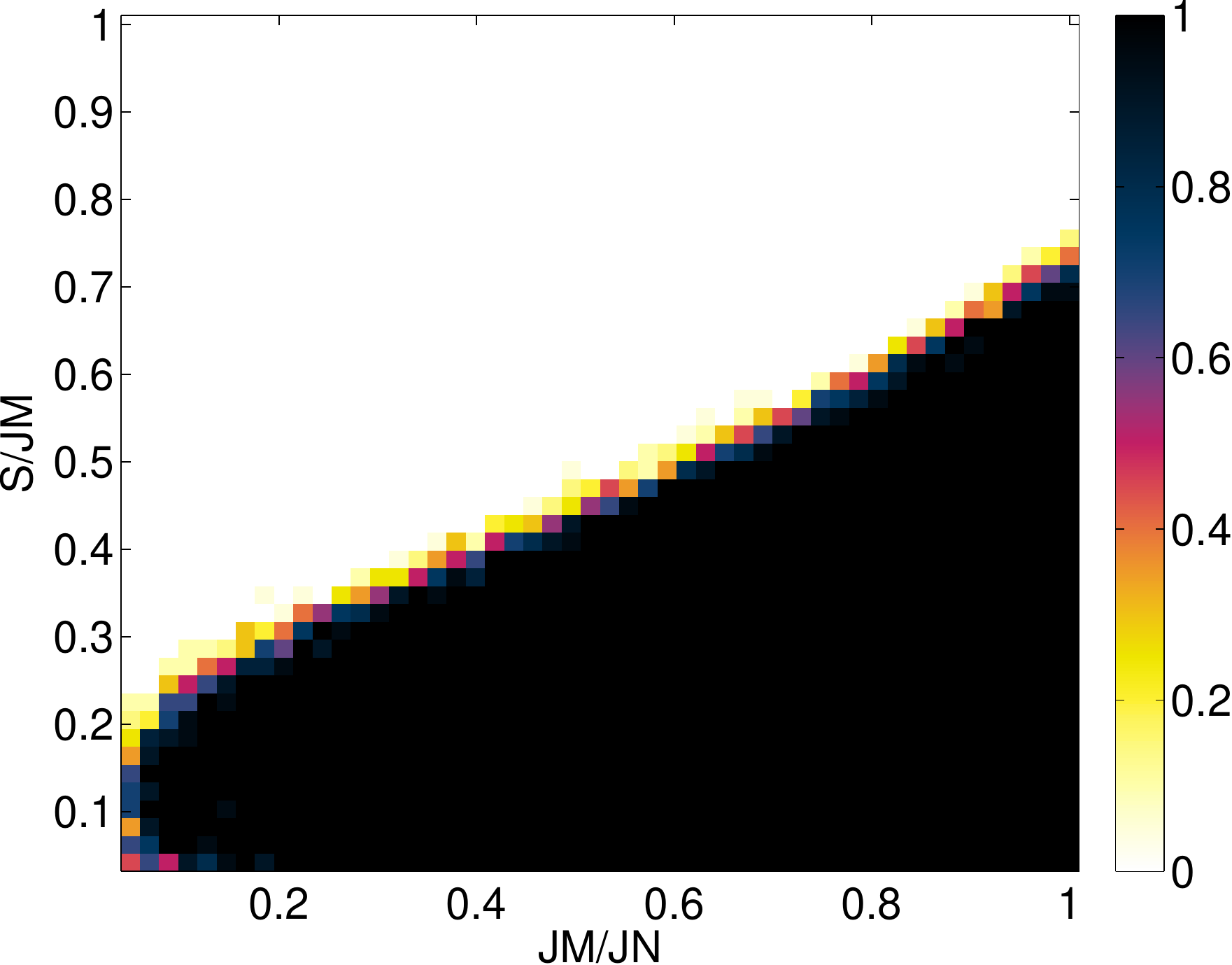}
	  \caption{Fourier basis}
	\end{subfigure}
	
	\begin{subfigure}[b]{0.45\linewidth}
	  \centering
	  \includegraphics[scale=.4]{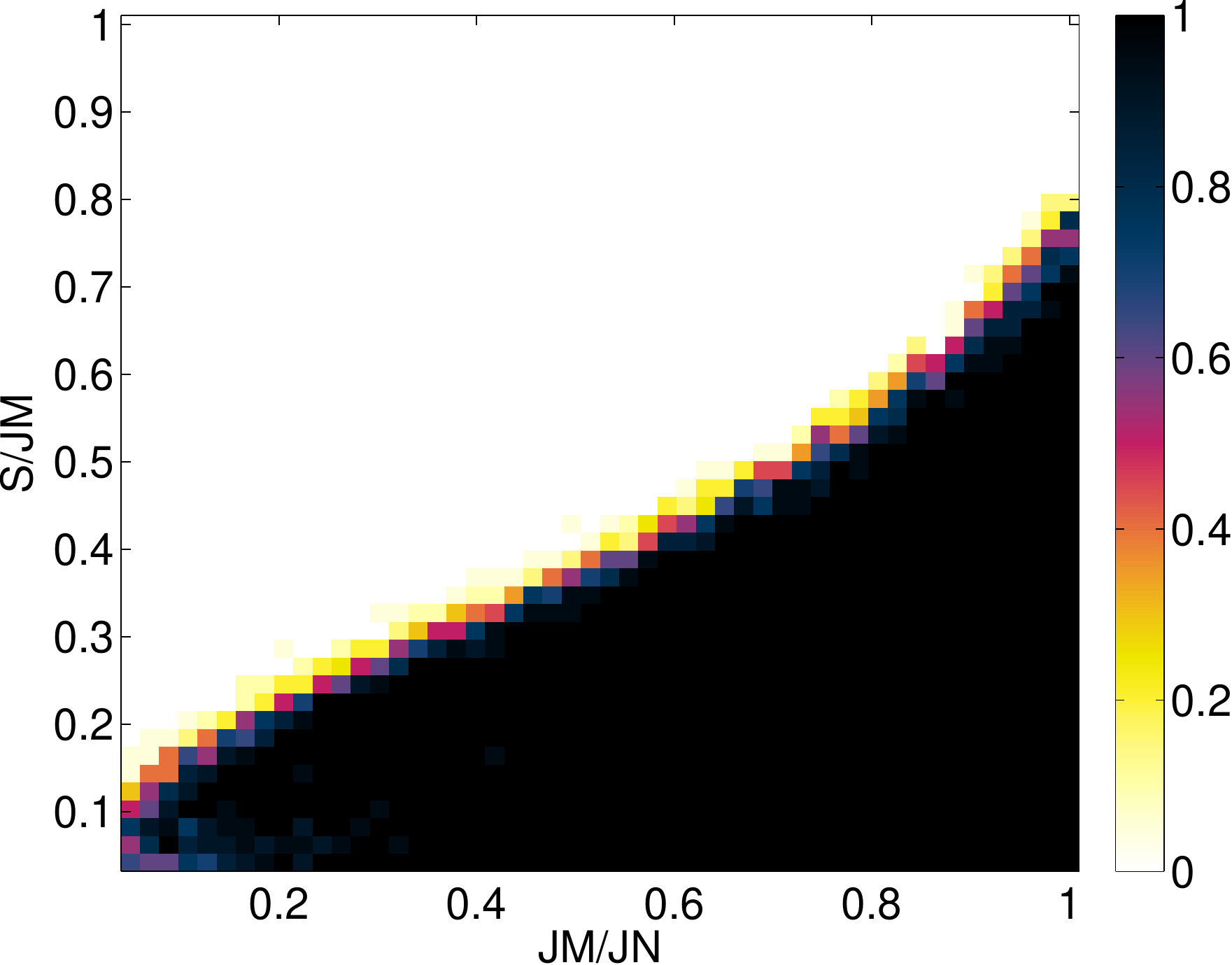}
	  \caption{Generic basis}
	\end{subfigure}

	\caption{Simulation results for DBD matrices.  Refer to Section~\ref{sec:simulations} for details.}

	\label{fig:DBD simulations}

\end{figure}

\begin{figure}[H]
	\centering
	\begin{subfigure}[b]{0.45\linewidth}
	  \centering
	  \includegraphics[scale=.4]{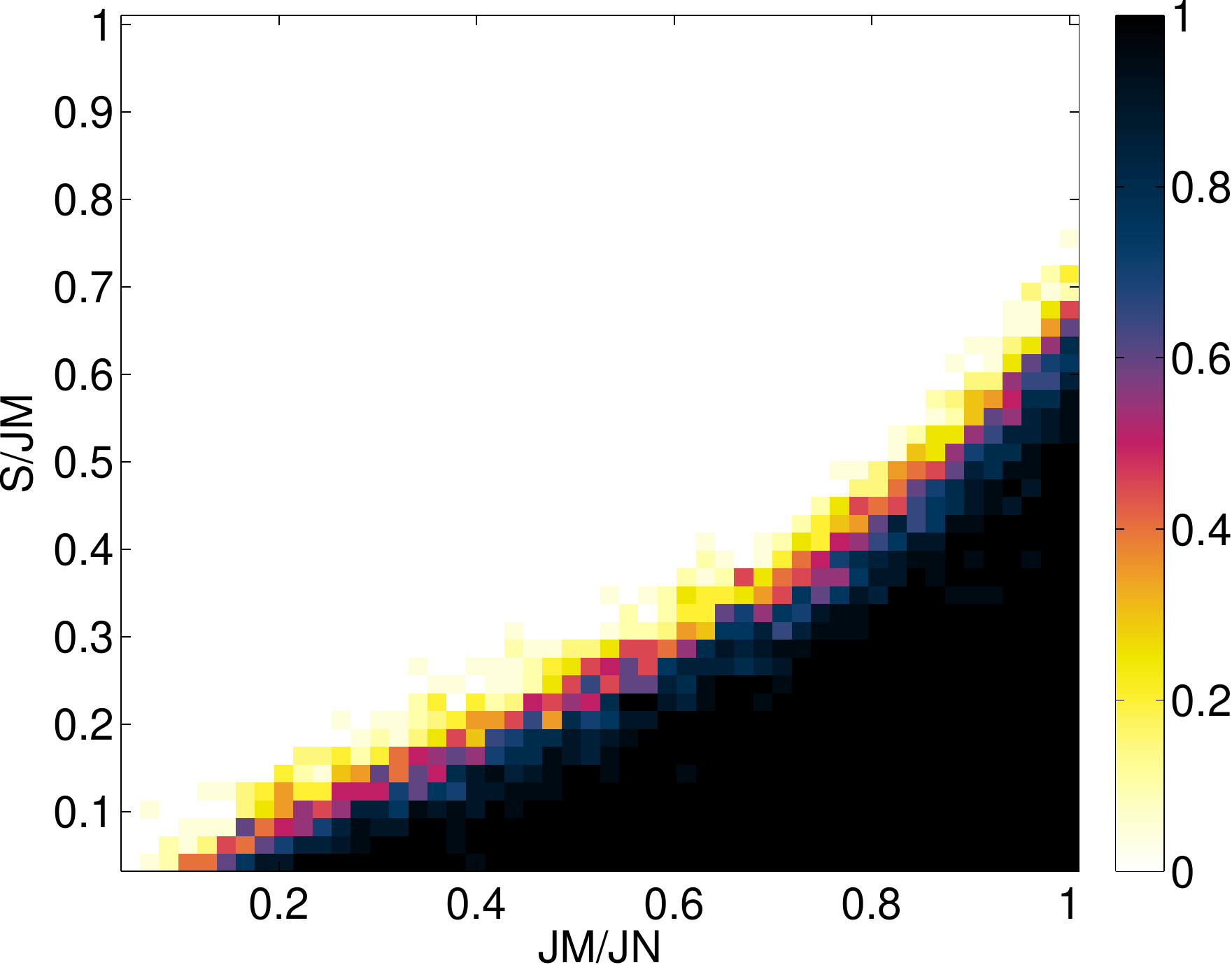}
	  \caption{Canonical basis}
	\end{subfigure}
	\begin{subfigure}[b]{0.45\linewidth}
	  \centering
	  \includegraphics[scale=.4]{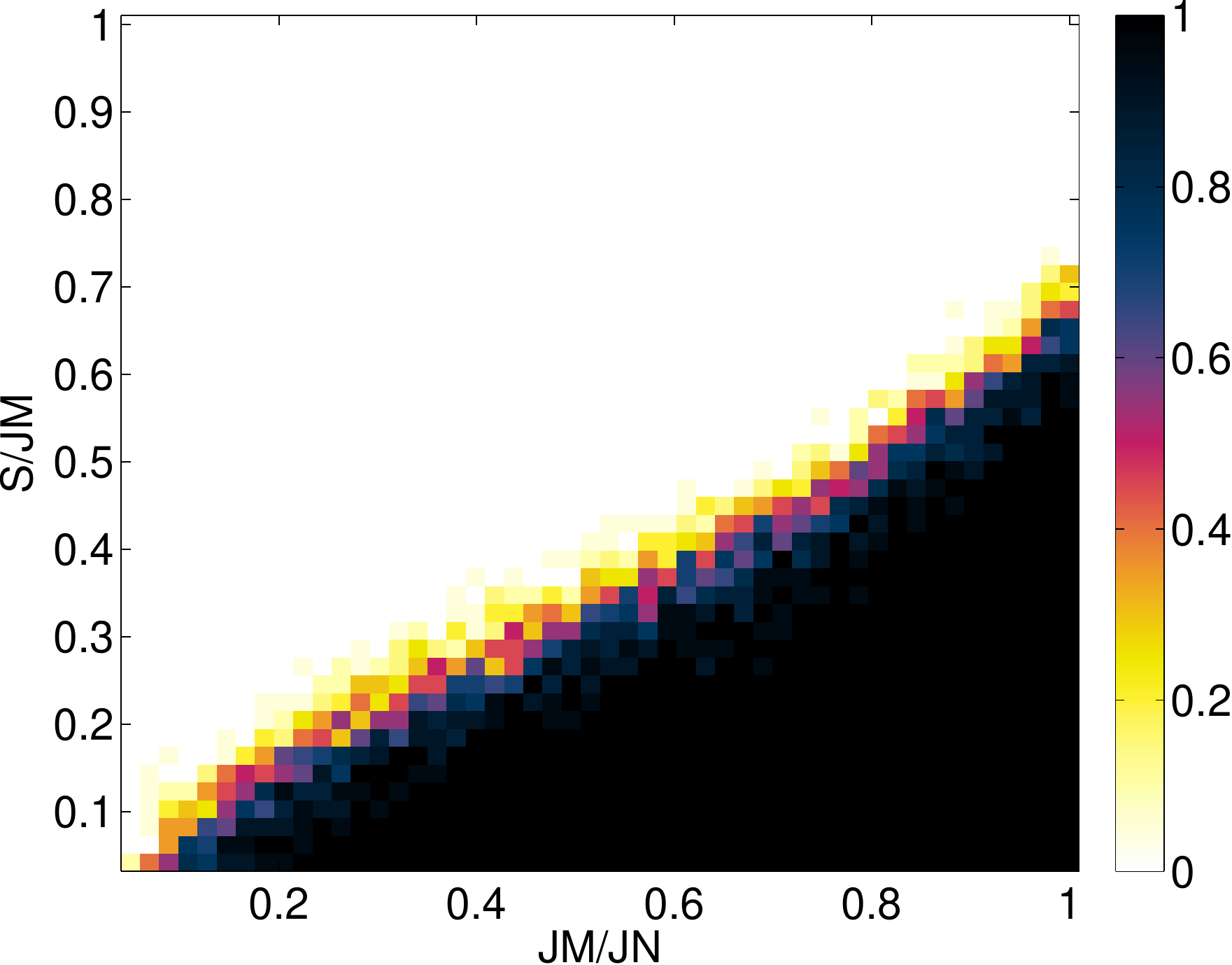}
	  \caption{Fourier basis}
	\end{subfigure}
	
	\begin{subfigure}[b]{0.45\linewidth}
	  \centering
	  \includegraphics[scale=.4]{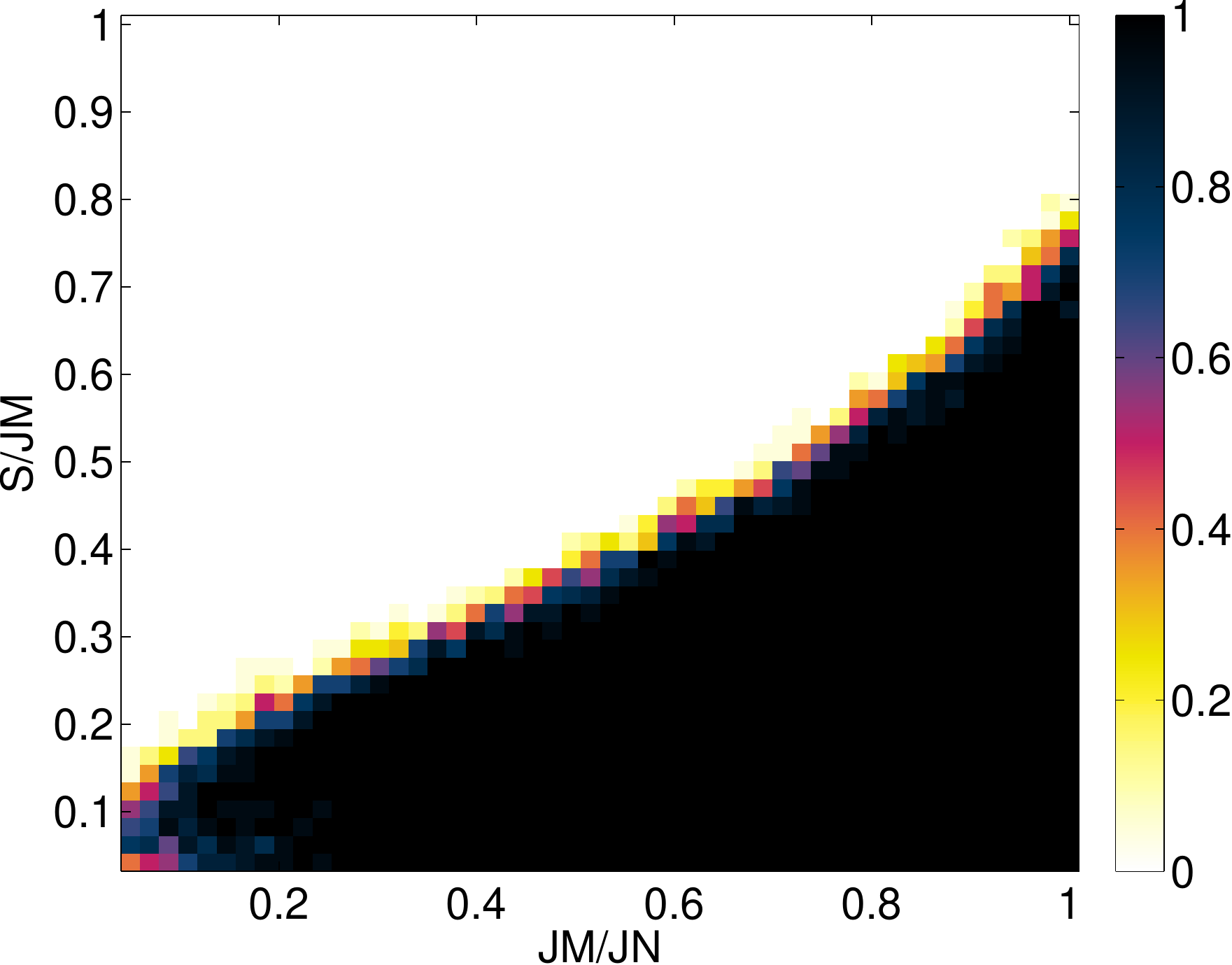}
	  \caption{Generic basis}
	\end{subfigure}

	\caption{Simulation results for RBD matrices. Refer to Section~\ref{sec:simulations} for details.}

	\label{fig:RBD simulations}

\end{figure}

\section{Proofs of Theorems~\ref{thm:DBD main thm} and~\ref{thm:RBD main thm}}
\label{sec:proof of DBD RBD main result}

\subsection{Preliminaries}

First, define the set of all $S$-sparse signals with unit norm as
\begin{equation}\label{eq:def of sparses}
	\sparses{S}:=\left\{\alpha\in\comps^{\bigN}\,:\,\zeronorm{\alpha}\le S,\,\twonorm{\alpha} = 1\right\}.
\end{equation}
With this definition, we observe that the RIC for DBD matrices can be written as
\begin{eqnarray*}
	\RIPcond{S} &=& \sup_{\alpha \in \sparses{S}}\left| \|\DBD \cdot x(\alpha)\|_2^2 - 1 \right|,
\end{eqnarray*}
where we leveraged the fact that $\twonorm{\alpha}=1$ implies
\begin{eqnarray*}
	\E{\|\DBD \cdot x(\alpha)\|_2^2}{} = \alpha^* \basis^* \E{\DBD^* \DBD}{} \basis \alpha = \alpha^* \alpha = 1.
\end{eqnarray*}
The RIC for RBD matrices can be similarly written as
\begin{eqnarray*}
	\RIPcond{S} &=& \sup_{\alpha \in \sparses{S}}\left| \|\RBD \cdot x(\alpha)\|_2^2 - 1 \right|.
\end{eqnarray*}
Given $\RIPcondgl<1$ and under the conditions in Theorems~\ref{thm:DBD main thm} and~\ref{thm:RBD main thm}, our objective is to show that $\RIPcond{S} \le \RIPcond{}$ for both DBD and RBD matrices. To achieve this goal, we require the following powerful result due to Krahmer et al.:
\begin{thm}\emph{\cite[Theorem 3.1]{Krahmer2012}}
	\label{thm:new Dudley}
	Let $\setofmat\subset\comps^{\overline{M}\times\overline{N}}$  be a set of matrices, and let $\rrad$ be a random vector whose entries are i.i.d.,\ zero-mean, unit-variance random variables with sub-Gaussian norm $\sgn$. Set
	\begin{align*}
		d_F(\setofmat) &:= \sup_{A \in \setofmat} \|A\|_F, \\
		d_2(\setofmat) &:= \sup_{A \in \setofmat} \|A\|_2,
	\end{align*}
	and
		\begin{eqnarray*}
		\parone &:=&\gammafcn_2\left(\setofmat, \| \cdot \|_2 \right) \left( \gamma_2\left(\setofmat, \| \cdot \|_2 \right) + d_F(\setofmat) \right) + d_F(\setofmat) d_2(\setofmat), \\
		\partwo &:=& d_2(\setofmat) \left( \gamma_2\left(\setofmat, \| \cdot \|_2 \right) + d_F(\setofmat) \right),\\
		\parthree &:=& d_2^2(\setofmat).
	\end{eqnarray*}
	Then, for $\dev > 0$, it holds that
	\begin{eqnarray*}
		\log\Proba{\sup_{A \in \setofmat} \left| \|A \rrad\|_2^2 - \Eb{\|A \rrad\|_2^2}{}\right| \gtrsim_\sgn \parone + \dev}{} \lesssim_\sgn - \min\left( \frac{\dev^2}{\partwo^2}, \frac{\dev}{\parthree} \right).
	\end{eqnarray*}
\end{thm}
Without going into the details, we note that the $\gammafcn_2$-functional of $\setofmat$, $\gammafcn_2\left(\setofmat, \twonorm{\cdot} \right)$, is a geometrical property of $\setofmat$, i.e., the index set of the random process, and is widely used in the theory of probability in Banach spaces \cite{Talagrand2005, ledoux1991probability}.\footnote{The $\gamma_2$-functional of $\mathcal{A}$, $\gamma_2(\mathcal{A},\|\cdot\|_2)$, is not to be confused with the block-coherence of an orthobasis $U$, $\gamma(U)$.} In particular, the following lemma gives an estimate of this quantity.
\begin{lemma}\emph{\cite{Talagrand2005}}
	\label{lemma:useful bnd on gamma2}
	With $\setofmat$ as defined above, it holds that
	\begin{equation}\label{eq:bound on gamma2}
		\gammafcn_2\left(\mathcal{A}, \| \cdot \|_2 \right) \lesssim \int_{0}^{\infty} \log^{\frac{1}{2}} \left(\covera{\setofmat}{\twonorm{\cdot}}{\dummy} \right)\,\,d\dummy.
	\end{equation}

\end{lemma}
Clearly, we need to express the problem of bounding the RIC of DBD and RBD matrices in a form that is amenable to the setting of Theorem~\ref{thm:new Dudley}. First, for DBD matrices, let us define $X_{D,j}\in\comps^{M\times MN}$, $j\in\compress{J}$, as
\begin{equation}
	X_{D,j}(\alpha)=X_{D,j}(\alpha,\basis):=\left[
	\begin{array}{cccc}
		x_j^*(\alpha) & & &\\
		& x_j^*(\alpha) & &\\
		& & \ddots &\\
		& & & x_j^*(\alpha)\\
	\end{array}
	\right].
\end{equation}
It can then be easily verified that
\begin{align*}
	\twonorm{\DBD \cdot x(\alpha)}^2 &= \sum_{j\in\compress{J}} \twonorm{\BD_j\cdot x_j(\alpha)}^2\\
	& = \sum_{j\in\compress{J}} \twonorm{X_{D,j}(\alpha)\cdot\vec(\BD_j^*)}^2\\
	& \underset{\mbox{\tiny i.d.}}{=} \sum_{j\in\compress{J}} \twonorma{\frac{1}{\sqrt{M}}X_{D,j}(\alpha)\cdot\rrad_j}^2\\
	& \underset{\mbox{\tiny i.d.}}{=} \twonorm{\A_D(\alpha)\cdot \rrad}^2,
\end{align*}
where the linear map $\A_D:\sparses{S}\rightarrow\comps^{\bigM\times JMN}$ is defined as
\begin{equation}\label{eq:def of AD}
	\A_D(\alpha) =\A_D(\alpha,\basis):= \frac{1}{\sqrt{M}}\left[
	\begin{array}{cccc}
		X_{D,1}(\alpha) & & &\\
		& X_{D,2}(\alpha) & &\\
		& & \ddots &\\
		& & & X_{D,J}(\alpha)\\
	\end{array}
	\right],
\end{equation}
and entries of $\rrad_j\in\reals^{MN}$, $j\in\compress{J}$, and $\rrad\in\reals^{JMN}$ are i.i.d.\ zero-mean, unit-variance random variables with sub-Gaussian norm $\sgn$. The index set of the random process is
\begin{equation}\label{eq:def of cal AD}
\setofmat_D:=\{\A_D(\alpha)\,:\,\alpha\in\sparses{S}\}.
\end{equation}
We have therefore completely expressed the DBD problem in the setting of Theorem~\ref{thm:new Dudley}, where $A$ and $\mathcal{A}$ are replaced with $A_D$ and $\mathcal{A}_D$, respectively.

Next, for RBD matrices, we observe that
\begin{align}
	\twonorm{\RBD \cdot x(\alpha)}^2 & =\sum_{j\in\compress{J}}\twonorm{\B \cdot x_j(\alpha)}^2\nonumber \\
	& = \frobnorm{\BD\cdot X_R(\alpha)}^2\nonumber\\
	& = \frobnorm{X_R^*(\alpha)\cdot\BD^*}^2\nonumber\\
	& \underset{\mbox{\tiny i.d.}}{=} \sum_{m\in\compress{M}} \twonorma{\frac{1}{\sqrt{M}}X_R^*(\alpha)\cdot\rrad_m'}^2\nonumber\\
	& \underset{\mbox{\tiny i.d.}}{=} \twonorm{\A_R(\alpha)\cdot\rrad'}^2,
\end{align}
where in the last line we defined  the linear map $\A_R:\sparses{S}\rightarrow\comps^{\bigM\times MN}$ as
\begin{equation}\label{eq:def of A_r}
	A_R(\alpha)=A_R(\alpha,\basis):=\frac{1}{\sqrt{M}}\left[
	\begin{array}{cccc}
		X_R^{*}\left(\alpha\right)\\
		& X_R^{*}\left(\alpha\right)\\
		&  & \ddots\\
		&  &  & X_R^{*}\left(\alpha\right)
	\end{array}\right],
\end{equation}
and the entries of $\rrad_m'\in\reals^{N}$, $m\in\compress{M}$, and $\rrad'\in\reals^{MN}$ are i.i.d.\ zero-mean, unit-variance random variables with sub-Gaussian norm of $\sgn$. Here, the index set of random process is
\begin{equation}\label{eq:def of cal AR}
\mathcal{A}_R:=\{A_R(\alpha)\,:\,\alpha\in\Omega_S\}.
\end{equation}
So, we have also managed to express the RBD problem in the setting of Theorem~\ref{thm:new Dudley} ,where $A$ and $\mathcal{A}$ are replaced with $A_R$ and $\mathcal{A}_R$, respectively.

The next two subsections are concerned with estimating the quantities involved in Theorem~\ref{thm:new Dudley} for both the DBD and RBD problems.

\subsection{Calculating $d_2(\setofmat_D)$, $d_F(\setofmat_D)$, and $\gammafcn_2(\setofmat_D,\|\cdot\|_2)$ }
\label{sec:calc quantities for DBD}

Recall the definitions of $A_D$ and $\mathcal{A_D}$ in \eqref{eq:def of AD} and \eqref{eq:def of cal AD}. We begin with defining the following norm on $\comps^{\bigN}$, which will find extensive use in the analysis of the DBD problem:
\begin{equation}\label{eq:def of X DBD norm}
	\|\alpha\|_{A_D}:=\|\A_D(\alpha)\|_2
\end{equation}
for $\alpha\in\comps^{\bigN}$. We record a useful property of this norm below.
\begin{lemma}\label{lemma:prop of X DBD norm}
	For every $\alpha\in\comps^{\bigN}$, it holds that
	\begin{equation}
		\|\alpha\|_{A_D}\le \frac{\modcoh}{\sqrt{\bigM}}\cdot \|\alpha\|_1.
	\end{equation}
\end{lemma}
\begin{proof}
	Let $\smallbasis_{j,n}$, $j\in\compress{J}$ and $n\in\compress{N}$, denote the ($(j-1)N+n$)th row of $\basis$. We then have that
	\begin{align*}
		\|\alpha\|_{A_D} & =\|\A_D(\alpha)\|_2\\
		& =\|\A_D(\alpha)\A_D^*(\alpha)\|_2^{\frac{1}{2}}\\
		&=\frac{1}{\sqrt{M}}\max_{j\in\compress{J}} \|x_j\|_2\\
		&=\frac{1}{\sqrt{M}}\max_{j\in\compress{J}} \|\basis_j\alpha\|_2\\
		& \le \sqrt{\frac{N}{M}}\max_{j\in\compress{J},n\in\compress{N}} \left|\left\langle \smallbasis_{j,n},\alpha \right\rangle\right|\\
		& \le \sqrt{\frac{N}{M}}\max_{j\in\compress{J},n\in\compress{N}} \|\smallbasis_{j,n}\|_\infty\|\alpha\|_1\\
		& = \frac{\coh}{\sqrt{\bigM}}\cdot\|\alpha\|_1,
	\end{align*}
where the second to last line uses the H\"{o}lder inequality and the last line follows from the definition of $\coh$. On the other hand, one may also write that 
\begin{equation*}
	\|\alpha\|_{A_D} =\frac{1}{\sqrt{M}}\max_{j\in\compress{J}} \|x_j\|_2 \le \frac{1}{\sqrt{M}}\|x\|_2=\frac{1}{\sqrt{M}}\|\basis\alpha\|_2=\frac{1}{\sqrt{M}}\|\alpha\|_2\le\frac{1}{\sqrt{M}}\|\alpha\|_1,
\end{equation*}
where we used the fact that $\basis$ is an orthobasis. Overall, we arrive at
\begin{equation}
	\|\alpha\|_{A_D}\le \frac{1}{\sqrt{\bigM}}\min\left(\coh,\sqrt{J}\right)\|\alpha\|_1=\frac{\modcoh}{\sqrt{\bigM}}\cdot\|\alpha\|_1,
\end{equation}
as claimed. The equality above follows from the definition of $\modcoh$ in \eqref{eq:def of mutilde}.
\end{proof}

We continue with computing the quantities involved in Theorem~\ref{thm:new Dudley} in the case of the DBD problem.
First, we have that
\begin{align}\label{eq:bnd on dF for DBD}
	d_F(\setofmat_D)
	= \sup_{\A_D(\alpha) \in \setofmat_D} \frobnorm{\A_D(\alpha)}
	= \sup_{\alpha \in \sparses{S}} \twonorm{x(\alpha)}
	= \sup_{\alpha \in \sparses{S}} \twonorm{\basis\alpha}
	= \sup_{\alpha \in \sparses{S}} \twonorm{\alpha}
	= 1.
\end{align}
The second to last equality holds because $\basis$ is an orthonormal basis. Second, we have that
\begin{align}\label{eq:bnd on d2 for DBD}
	d_2(\setofmat_D)
	= \sup_{\A_D(\alpha) \in \setofmat_D} \|\A_D(\alpha)\|_2
	= \sup_{\alpha \in \sparses{S}} \|\alpha\|_{A_D}\le \frac{\modcoh}{\sqrt{\bigM}}\sup_{\alpha \in \sparses{S}}\|\alpha\|_1\le \modcoh\sqrt{\frac{S}{\bigM}}.
\end{align}
The first inequality above holds on account of Lemma~\ref{lemma:prop of X DBD norm}. The second inequality above follows because $\|\alpha\|_2=1$ and $\|\alpha\|_0\le S$ when $\alpha\in\sparses{S}$.
It is only left to bound $\gammafcn_2(\setofmat_D, \|\cdot\|_2)$. According to Lemma~\ref{lemma:useful bnd on gamma2}, we have that
\begin{align*}
	\gammafcn_2(\setofmat_D, \|\cdot\|_2) & \le \int_{0}^{\infty} \log^{\frac{1}{2}} \left(\Cover\left(\setofmat_D,\|\cdot \|_2,\dummy\right)\right) \,\,d\dummy\\
	& =\int_{0}^{\infty} \log^{\frac{1}{2}}\left( \Cover\left(\sparses{S},\|\cdot \|_{A_D},\dummy\right)\right) \,\,d\dummy,
\end{align*}
where the isometry between $\setofmat_D$ (with metric $\|\cdot\|_2$) and $\sparses{S}$ (with metric $\|\cdot\|_{A_D}$) implies the second line. This isometry, in turn, follows from~\eqref{eq:def of X DBD norm} and the linearity of $\A_D(\cdot)$.  Consequently,
\begin{align}
	\gamma_2(\setofmat_D, \|\cdot\|_2)
	&\le \int_{0}^{\infty} \log^{\frac{1}{2}} \left( \covera{\frac{\sparses{S}}{\sqrt{S}}}{\|\cdot \|_{A_D}}{\frac{\dummy}{\sqrt{S}}} \right) d\dummy \nonumber\\
	&\le \sqrt{S} \int_{0}^{\infty} \log^{\frac{1}{2}} \left( \covera{\frac{\sparses{S}}{\sqrt{S}}}{\|\cdot \|_{A_D}}{\dummy} \right) d\dummy, \label{eq:mu 2 calculate cover}
\end{align}
where the first line uses the second inequality in~\eqref{eq:conv fact about covering} and the last line follows from a change of variables in the integral. An estimate of the covering number involved in~\eqref{eq:mu 2 calculate cover} can be found through the next result, which is proved in \ref{sec: main cover no RBD}.
\begin{lemma} \label{lem: main cover no}
	Consider a norm $\|\cdot\|_A$ on $\comps^{\bigN}$ that, for every $\alpha\in\comps^{\bigN}$, satisfies
	\begin{equation*}
		\|\alpha\|_A=\|A(\alpha)\|_2\le \frac{\Qall}{\sqrt{\bigM}}\cdot\|\alpha\|_1,
	\end{equation*}
	for some linear map $A(\cdot):\comps^{\bigN}\rightarrow\comps^{N'}$ with rank of at most $\bigM$ and some $\Qall>0$ and integer $N'$.
        Then, for $0< \dummy < \Qall/\sqrt{\bigM}$ and $\bigM\gtrsim1$, we have that
	\begin{equation}
		\log\left(\covera{\frac{\sparses{S}}{\sqrt{S}}}{\|\cdot\|_A}{\dummy}\right)  \lesssim
		\min\left(S\log\bigN+S\log\left(1+\frac{2\Qall}{\dummy\sqrt{\bigM}}\right),\frac{\Qall^2}{\dummy^2\bigM}\cdot\log^2\bigN\right).
		\label{eq:log of cov no for norm X}
	\end{equation}
	When $\dummy \ge \Qall/\sqrt{\bigM}$, we have $\covera{\frac{\sparses{S}}{\sqrt{S}}}{\|\cdot\|_A}{\dummy} = 1$.
\end{lemma}
Qualitatively speaking, of the two bounds on the right hand of \eqref{eq:log of cov no for norm X}, the first is tighter when $\dummy$ is small while the second is more effective for larger values of $\dummy$. Of course, $\|\cdot\|_{A_D}$ satisfies the hypothesis of Lemma~\ref{lem: main cover no} with $\Qall=\modcoh$ and the map $A_D(\cdot)$. Consequently, for $0<\thresh\le\modcoh/\sqrt{\bigM}$ to be set later, we have that
\begin{align}
	  & \int_{0}^\infty \log^{\frac{1}{2}} \left(\covera{\frac{\sparses{S}}{\sqrt{S}}}{\|\cdot\|_{A_D}}{\dummy} \right)\, d\dummy \nonumber\\
	  & =\int_{0}^{\thresh} \log^{\frac{1}{2}} \left(\covera{\frac{\sparses{S}}{\sqrt{S}}}{\|\cdot\|_{A_D}}{\dummy} \right)\, d\dummy
	+ \int_{\thresh}^{\frac{\modcoh}{\sqrt{\bigM}}} \log^{\frac{1}{2}} \left(\covera{\frac{\sparses{S}}{\sqrt{S}}}{\|\cdot\|_{A_D}}{\dummy} \right)\, d\dummy \nonumber \\
	  & \lesssim \int_{0}^{\thresh} \left(\sqrt{S\log\bigN} + \sqrt{S\log\left(1+\frac{2\modcoh}{\dummy\sqrt{\bigM}}\right)}\right)\, d\dummy+\log\bigN\int_{\thresh}^{\frac{\modcoh}{\sqrt{\bigM}}}\frac{\modcoh}{\dummy\sqrt{\bigM}}\, d\dummy\nonumber \\
	  & \lesssim \thresh\sqrt{S \log\bigN} + \thresh\sqrt{S \log\left(1+\frac{2\modcoh}{\thresh\sqrt{\bigM}}\right)} + \frac{\modcoh}{\sqrt{\bigM}}\log\bigN\log\left(\frac{\modcoh}{\thresh\sqrt{\bigM}}\right) .
	\label{eq:the subgauss integral}
\end{align}
The second line above follows from the second statement in Lemma~\ref{lem: main cover no}. In the third line, different upper bounds from \eqref{eq:log of cov no for norm X} are used to bound each summand. We benefited from \eqref{eq:log integral 2} in the Toolbox to compute the logarithmic integral in the third line.  With the choice of $\thresh = \modcoh/\sqrt{S\bigM} $, we obtain that
\begin{align}	\label{eq:mid step subg mu}
	& \int_{0}^\infty \log^{\frac{1}{2}} \left(\covera{\frac{\sparses{S}}{\sqrt{S}}}{\|\cdot\|_{A_D}}{\dummy} \right)\, d\dummy \nonumber\\
	&\lesssim  \frac{\modcoh}{\sqrt{\bigM}} \sqrt{\log\bigN} + \frac{\modcoh}{\sqrt{\bigM}} \sqrt{\log(1+2\sqrt{S})} + \frac{\modcoh}{\sqrt{\bigM}} \log S \log \bigN \nonumber\\
	&\lesssim \frac{\modcoh}{\sqrt{\bigM}} \log S\log\bigN
\end{align}
for $S\gtrsim1$. Now plugging back~\eqref{eq:mid step subg mu} into~\eqref{eq:mu 2 calculate cover}, we arrive at
\begin{equation}\label{eq:bnd on gamma2 for DBD}
	\gammafcn_2(\setofmat_D, \|\cdot\|_2) \lesssim \modcoh\sqrt{\frac{S}{\bigM}} \log S\log\bigN.
\end{equation}
Before completing the analysis of the DBD problem, let us calculate the same quantities for the RBD case.

\subsection{Calculating $d_2(\setofmat_R)$, $d_F(\setofmat_R)$, and $\gammafcn_2(\setofmat_R,\|\cdot\|_2)$ }
\label{sec:calc quantities for RBD}

Recall the definitions of $A_R$ and $\mathcal{A_R}$ in \eqref{eq:def of A_r} and \eqref{eq:def of cal AR}. Again, we first introduce the following norm on $\comps^{\bigN}$, which will be useful in the analysis of the RBD problem:
\begin{equation}\label{eq:def of X RBD norm}
	\|\alpha\|_{A_R}:=\|\A_R(\alpha)\|_2
\end{equation}
for $\alpha\in\comps^{\bigN}$. This norm has the following property.
\begin{lemma}\label{lemma:prop of X RBD norm}
	For every $\alpha\in\comps^{\bigN}$, it holds that
	\begin{equation}
		\|\alpha\|_{A_R}\le \frac{\Q}{\sqrt{\bigM}}\cdot \|\alpha\|_1.
	\end{equation}
\end{lemma}
\begin{proof}
	Note that
	\begin{align*}
		\|\alpha\|_{A_R} & = \|A_R(\alpha)\|_2\\
		&=\frac{1}{\sqrt{M}}\|X_R(\alpha)\|_2 \nonumber\\
		& =\frac{1}{\sqrt{M}}\twonorm{\sum_{\bign\in\compress{\bigN}}\alpha(\bign)X_R(e_{\bign})}\nonumber \\
		& \le\frac{1}{\sqrt{M}}\sum_{\bign\in\compress{\bigN}}\abs{\alpha(\bign)}\cdot\twonorm{X_R(e_{\bign})}\nonumber \\
		& \le\frac{1}{\sqrt{M}}\max_{\bign\in\compress{\bigN}}\twonorm{X_R(e_{\bign})}\cdot\onenorm{\alpha}\nonumber \\
		& =\frac{\Q}{\sqrt{\bigM}} \cdot\onenorm{\alpha}.
	\end{align*}
	The third line above uses the linearity of $X_R(\cdot)$. The fourth line  follows from the triangle inequality and the fifth line is implied by the H\"{o}lder inequality. We made use of the definition of $\Q$ in \eqref{eq:def of block-coherence} to get the last line.
\end{proof}

We now continue with computing the quantities involved in Theorem~\ref{thm:new Dudley} in the case of the RBD problem. First, we have that
\begin{align}
	d_F(\setofmat_R)= \sup_{A_R(\alpha) \in \setofmat_R} \|A_R(\alpha)\|_F
        = \sup_{\alpha \in \sparses{S}} \|X_R(\alpha)\|_F
	= \sup_{\alpha \in \sparses{S}} \|x(\alpha)\|_2
	=  \sup_{\alpha \in \sparses{S}} \|\alpha\|_2
	= 1.\label{eq:bnd on dF for RBD}
\end{align}
Second, we have that
\begin{align}\label{eq:bnd on d2 for RBD}
	d_2(\setofmat_R)
	= \sup_{A_R(\alpha) \in \setofmat_R} \|A_R(\alpha)\|_2
	= \sup_{\alpha \in \sparses{S}} \|\alpha\|_{A_R}
	\le \frac{\Q}{\sqrt{\bigM}}\sup_{\alpha \in \sparses{S}} \|\alpha\|_1
	\le \Q\sqrt{\frac{S}{\bigM}},
\end{align}
where the first inequality follows from Lemma~\ref{lemma:prop of X RBD norm}.
The last quantity to estimate is $\gamma_2(\setofmat_R, \|\cdot\|_2)$. As was the case in the previous subsection, we may write that
\begin{equation}
	\gamma_2(\setofmat_R, \|\cdot\|_2)
	\le \sqrt{S} \int_{0}^{\infty} \log^{\frac{1}{2}} \left( \covera{\frac{\sparses{S}}{\sqrt{S}}}{\|\cdot \|_{A_R}}{\dummy} \right) d\dummy. \label{eq:mu 2 calculate cover RBD}
\end{equation}
Of course, $\|\cdot\|_{A_R}$ satisfies the hypothesis of Lemma~\ref{lem: main cover no} with $\Qall=\Q$ and the map $A_R(\cdot)$. Therefore, following the same steps as in the previous subsection, we arrive at
\begin{eqnarray}\label{eq:bnd on gamma2 for RBD}
	\gamma_2(\setofmat_R, \|\cdot\|_2) \lesssim \Q\sqrt{\frac{S}{\bigM}} \log S\log\bigN.
\end{eqnarray}

\subsection{Denouement}

We notice that the quantities $d_F(\setofmat_D), d_2(\setofmat_D),$ and $\gamma_2(\setofmat_D,\|\cdot\|_2)$ have the same bounds as their counterparts $d_F(\setofmat_R), d_2(\setofmat_R),$ and $\gamma_2(\setofmat_R,\|\cdot\|_2)$ except for the type of the coherence factor involved. Therefore, it suffices to focus on one---the same result holds for the other with its corresponding coherence factor. Given $\RIPcond{} < 1$, assume that
\begin{eqnarray}\label{eq:bnd on Mbig}
	\bigM \gtrsim_\sgn \RIPcond{}^{-2}\modcoh^2 \cdot S \log^2 S\log^2 \bigN.
\end{eqnarray}
Equipped with the estimates in Section~\ref{sec:calc quantities for DBD}, i.e., \eqref{eq:bnd on dF for DBD}, \eqref{eq:bnd on d2 for DBD}, and \eqref{eq:bnd on gamma2 for DBD}, we now compute $\parone$ in Theorem~\ref{thm:new Dudley}:
\begin{eqnarray*}
	\parone &:=& \gamma_2\left(\setofmat_D, \| \cdot \|_2 \right) \left( \gamma_2\left(\setofmat_D, \| \cdot \|_2 \right) + d_F(\setofmat_D) \right) + d_F(\setofmat_D) d_2(\setofmat_D) \\
	&\lesssim& \modcoh \sqrt{\frac{S}{\bigM}}\log S\log \bigN \left( \modcoh \sqrt{\frac{S}{\bigM}}\log S\log \bigN + 1 \right) + \modcoh\sqrt{\frac{ S}{\bigM}} \\
	&\lesssim_\sgn& \RIPcond{} \left( \RIPcond{} + 1 \right) + \frac{\RIPcond{}} {\log S\log \bigN} \\
	& \le &  2\RIPcond{} + \frac{\RIPcond{}} {\log S\log \bigN} \\
	&\lesssim& \RIPcond{},
\end{eqnarray*}
where we assumed that $S\gtrsim1$ and used the hypothesis that $\RIPcond{}< 1$ in the second to last line.
Doing the same to $\partwo$, we obtain that
\begin{eqnarray*}
	\partwo &:=& d_2(\setofmat_D) \left( \gamma_2\left(\setofmat_D, \| \cdot \|_2 \right) + d_F(\setofmat_D) \right) \\
	&\lesssim& \modcoh\sqrt{\frac{ S}{\bigM}} \left( \modcoh \sqrt{\frac{S}{\bigM}}\log S\log \bigN + 1 \right)\\
	&\lesssim_\sgn& \frac{\RIPcond{}}{ \log S\log \bigN} \left( \RIPcond{} + 1 \right) \\
	&\lesssim& \frac{\RIPcond{}}{ \log S\log \bigN}.
\end{eqnarray*}
Similarly for $\parthree$, we write that
\begin{eqnarray*}
	\parthree := d_2^2(\setofmat_D)
	\le \frac{\modcoh^2 S}{\bigM}
	\lesssim_\sgn \frac{\RIPcond{}^2}{\log^2 S\log^2 \bigN}.
\end{eqnarray*}
Plugging the above estimates of $\parone$, $\partwo$, and $\parthree$ into the tail bound in Theorem~\ref{thm:new Dudley}, we obtain that
\begin{equation*}
	\log\Proba{\sup_{\alpha \in \sparses{S}} \left| \|\DBD \cdot x(\alpha)\|_2^2 - 1 \right| \gtrsim_\sgn \RIPcond{} + t}{}
	\lesssim_\sgn - \min\left( \RIPcond{}^{-2}t^2 \log^2 S\log^2 \bigN, \RIPcond{}^{-1}t \log^2 S\log^2 \bigN \right) .
\end{equation*}
Substituting $\dev =  \RIPcond{}$, we arrive at
\begin{equation*}
	\log\Proba{\sup_{\alpha \in \sparses{S}}\left| \|\DBD\cdot x(\alpha)\|_2^2 - 1 \right| \gtrsim_\sgn \RIPcond{}}{} \lesssim_\sgn -\log^2 S\log^2\bigN,
\end{equation*}
assuming that $S\gtrsim 1$. After absorbing the factor depending on $\sgn$ into (a redefined) $\RIPcond{}$, we finally arrive at
\begin{align*}
	& \log\Proba{\sup_{\alpha \in \sparses{S}}\left| \|\DBD\cdot x(\alpha)\|_2^2 - 1 \right| > \RIPcond{}}{} \lesssim_\sgn-\log^2 S\log^2\bigN,
\end{align*}
which completes the proof of Theorem~\ref{thm:DBD main thm}. Replacing $\modcoh$ with $\Q$ and repeating this argument concludes the proof of Theorem~\ref{thm:RBD main thm}.

\section{Conclusion}
\label{sec:conclusion}
In this paper, we studied two  classes of structured random matrices, namely DBD and RBD matrices.
Our main results state that matrices with block diagonal constructions can indeed satisfy the RIP but that the requisite number of measurements depends on certain properties of the sparsifying basis. These properties were detailed and  interpreted in the paper.
In the best case, DBD and RBD matrices perform nearly as well as the dense i.i.d.\ random matrices  generally used in CS despite having many fewer nonzero entries. Moreover, we have shown that random block diagonal matrices are intimately related to the random convolution and random Toeplitz matrices considered in the literature.

Our findings lead us to conclude that structured random matrices can be useful in sensing architectures as long as the statistics of the data are well-matched with the structure of the measurement matrix.  While this intuition is similar to other results on structured measurement matrices, our results on block diagonal matrix constructions are novel in extending this intuition to matrices with (potentially) many entries that are zero.  The approach required to reach our results also leads us to conclude that future progress in
the field of probability in Banach spaces is likely to play a significant role in establishing optimal performance guarantees for other structured measurement systems.  This may be especially true given the improved performance we were able to achieve even over other bounding techniques that require sophisticated mathematical machinery. We do however remain uncertain about the necessity of the poly-logarithmic factors in the final measurement rates (which may be a proof artifact). Finally, as a means of visualizing how various choices for the sparsity basis will interact with a block diagonal measurement matrix, we study the empirical performance of signal recovery with random block diagonal measurement matrices. The simulation results show that the same bases which are most favorable for our RIP bounds tend to also be most favorable for empirical signal recovery performance.

There are several directions that can be explored in the future. First, as we have discussed in the introduction, block diagonal matrices are useful for modeling distributed measurement systems. It may therefore be of interest to specialize our results to some particular distributed systems. Take for example MIMO radar where multiple independent transmitters and receivers are arbitrarily distributed over an area of interest to sense targets in a scene. Data from the receivers with potentially high data rates needs to be sent to a central processor and to be coherently processed to achieve a maximum processing gain. The block diagonal structure studied here is potentially useful to analyze the possibility of compressing the data at the individual receivers before sending it to the central processor. For another example, block diagonal structure has also been exploited in observability studies of dynamical systems~\cite{wakin2010observability}. Our understanding of this and similar problems~\cite{Asif2011,
Charles2011,Charles2013} may be enhanced using the results in this paper.

\section*{Acknowledgements}
We would like to thank Borhan Sanandaji for helpful comments on an early version of the manuscript. We also gratefully acknowledge Justin Romberg and Alejandro Weinstein for valuable discussions and insightful comments. The authors would also like to thank the anonymous reviewers for their constructive comments and positive feedback. 

\bibliographystyle{plain}
\bibliography{BDreferences}

\appendix

\section{Toolbox}\label{sec:toolbox}

This section collects a few general results that are used throughout the paper (mainly without proofs, for the benefit of the space).

Schatten norms possess the following useful property that mirrors Euclidean norms:
\begin{equation}\label{eq:prop of schatten norm}
	\rank{\tempmat}^{\frac{1}{p}-\frac{1}{q}}\schatten{\tempmat}{q}\le\schatten{\tempmat}{p}\le\schatten{\tempmat}{q},
\end{equation}
for a matrix $\tempmat$, when $1\le q\le p$. The following version of the H\"{o}lder inequality for matrices is used in this paper. For any pair of matrices $\tempmat,\tempmatp$ (such that $\tempmat\tempmatp$ exists), the following holds:\footnote{
	Let $\{b_j\}$ denote the columns of $B$. Then $\|AB\|_F^2 = \sum_j \|A b_j\|_2^2 \le \sum_j \|A\|_2^2 \|b_j\|_2^2 = \|A\|_2^2 \|B\|_F^2$.
}

\begin{equation}\label{eq:holder for matrices}
\frobnorm{\tempmat\tempmatp}\le\twonorm{A}\frobnorm{B}.
\end{equation}
For a random variable $\rv$ and $1\le p\le q$, the following holds~\cite[Page 30]{rauhut2010compressive}:
\begin{equation}\label{eq:prop of Ep norm}
	\Eb{\abs{\rv}}{p}\le \Eb{\abs{\rv}}{q}.
\end{equation}
Also, $\sgnorm{\const\cdot\rv}=\abs{\const}\sgnorm{\rv}$ and, according to~\cite[Lemma 5.9]{vershynin2010introduction}, the following holds for a finite sequence of zero-mean independent random variables $\{\rv_j\}$:
\begin{equation}\label{eq:subg norm of sum of rvs}
	\sgnorm{\sum_j\rv_j}^2\lesssim \sum_j\sgnorm{\rv_j}^2.
\end{equation}

Throughout this section, let $\rgvec\in\reals^{N}$ denote a vector whose entries are i.i.d.\ zero-mean unit variance sub-Gaussian random variables. For convenience, set $\maxsgnorm:=\max_{n\in\compress{N}}\sgnorm{\rgvec(n)}$.
For $\dev>0$ and $n\in\compress{N}$, the following holds by the definition of the sub-Gaussian norm~\cite{vershynin2010introduction}:
\begin{equation}\label{eq:bound on signel gaussian rv}
	\log\Proba{\abs{\rgvec(n)}>\dev}{}\lesssim -\frac{\dev^2}{\sgnorm{\rgvec(n)}^2}.
\end{equation}
Also, the following holds when $N\gtrsim 1$,~\cite[Lemma 6.6]{vershynin2010introduction}:
\begin{equation}\label{eq:exp of max of gaussians}
	\sqrt{\Eb{ \|\rgvec\|^2_{\infty}}{}}=(\Eb{\max_{n\in\compress{N}}\rgvec^2(n)}{})^{1/2} \lesssim\maxsgnorm\sqrt{\log N}.
\end{equation}
Furthermore, for $\dev\le1$, the following inequality is from 
\cite[Corollary 5.17]{vershynin2010introduction} and provides a lower bound on $\twonorm{\rgvec}$:
\begin{equation}
	\log\Proba{\twonorm{\rgvec}<(1-\dev)\sqrt{N}}{} \lesssim -\min(\frac{\dev^2}
	{\maxsgnorm^4},\frac{\dev}{\maxsgnorm^2})N.
	\label{eq:bound on norm of Gaussian vector}
\end{equation}
Suppose that the entries of $\rgmat\in\reals^{N\times J}$ are Gaussian random variables with zero-mean and unit variance. The next inequality provides an upper bound on the spectral norm of $\rgmat$, which directly follows from Corollary 5.35 in~\cite{vershynin2010introduction}. When $J\le N$, and for $\dev>0$, the following holds:
\begin{equation}
	\log\Proba{\twonorm{\rgmat}>(1+\dev)\sqrt{N}+\sqrt{J}}{} \le -\dev^2N/2.
\label{eq:bound on norm of Gaussian matrix}
\end{equation}
The next result essentially bounds the moments of a sum of independent random variables with those of a Rademacher sequence. The proof of this result (and the next one) uses the symmetrization technique, which has the following argument at its heart. If $\rv$ is a random variable taking values in a Banach space $\Banach$, then we can define its symmetrized version $\symrv = \rv - \rv'$, where $\rv'$ is an independent copy of $\rv$. As suggested by the name, the distribution of $\symrv$ is symmetric about the origin. In addition, the distributions of $\symrv$ and $\rad \symrv$ are the same, where $\rad$ is a standard Bernoulli random variable that takes $\pm 1$ with equal probability.
\begin{lemma}\label{symmetrization in expectation}
\emph{\cite[Lemma 6.7]{rauhut2010compressive}}
	\label{lem:sym}
	Let $\{\rv_l\}$ be a finite sequence of independent random variables in a Banach space $(\Banach,\Banachnorm{\cdot})$. Then, the following holds: 
	\begin{eqnarray*}
		\Eb{\Banachnorm{\sum_{l} \rv_l - \Eb{\rv_l}{}}}{} \le 2 \Eb{\Banachnorm{\sum_{l} \rad_l Z_l}}{},
	\end{eqnarray*}
	where $\{\xi_l\}$ is a Rademacher sequence.
\end{lemma}



In our proofs, we also require a (weak version of) the Khintchine inequality for operator norms that we state next.
\begin{lemma}\label{lem:weak kintchine ineq}
	If $\{\tempmat_{l}\}$, $l\in\compress{\copies}$, is a sequence of  matrices of the same dimension and rank of at most $J \gtrsim 1$, then the following holds:
	\begin{equation*}
	\Eb{\twonorma{\sum_{l\in\compress{\copies}}\rad_{l}\tempmat_l}}{}
	\lesssim \sqrt{\log J}\left(\sum_{l\in\compress{\copies}}\twonorma{ \tempmat_l}^{2}\right)^{1/2},
	\end{equation*}
where $\{\xi_l\}$ is a Rademacher sequence.
\end{lemma}
\begin{proof}
From~\cite[Theorem 6.14]{rauhut2010compressive} and for every $2\le p<\infty$, we recall that
\[
	\Eb{\schattena{\sum_{l\in\compress{\copies}}\rad_{l}\tempmat_{l}}{p}}{p}	\lesssim\sqrt{p}\cdot\max\left(\schattena{(\sum_{l\in\compress{\copies}}\tempmat_l\tempmat_l^{*})^{\frac{1}{2}}}{p},\schattena{ (\sum_{l\in\compress{\copies}}\tempmat_l^{*}\tempmat_l)^{\frac{1}{2}}}{p}\right).
\]
The spectral norm is a special case of the Schatten norm with $p=\infty$. Therefore, the inequality above does not directly apply to our problem. As such, we need a more detailed argument here, which follows the approach of~\cite{vershynin2010introduction}. 
From \eqref{eq:prop of schatten norm}, with $p=\infty$ and $q=\log J$, recall that
\[
	e^{-1}\schatten{\tempmat}{\log J}\le\twonorm{\tempmat}\le\schatten{\tempmat}{\log J}
\]
for $J\ge e$.
This equivalence, in combination with the fact that $\Eb{}{p}$ is increasing in $p$, i.e., \eqref{eq:prop of Ep norm}, allows us to write
\begin{align*}
	\Eb{\twonorma{\sum_{l\in\compress{\copies}}\rad_{l}\tempmat_l}}{}& \le\Eb{\schattena{\sum_{l\in\compress{\copies}}\rad_{l}\tempmat_l}{\log J}}{}\\
	& \le\Eb{\schattena{\sum_{l\in\compress{\copies}}\rad_{l}\tempmat_l}{\log J}}{\log J}\\
	& \lesssim \sqrt{\log J}\max\left(\schattena{ (\sum_{l\in\compress{\copies}}\tempmat_l\tempmat_l^{*})^{\frac{1}{2}}}{\log J},\schattena{ (\sum_{l\in\compress{\copies}}\tempmat_l^{*}\tempmat_l)^{\frac{1}{2}}}{\log J}\right)\\
	& \le e\sqrt{\log J}\max\left(\twonorma{ (\sum_{l\in\compress{\copies}}\tempmat_l\tempmat_l^{*})^{\frac{1}{2}}} ,\twonorma{ (\sum_{l\in\compress{\copies}}\tempmat_l^{*}\tempmat_l)^{\frac{1}{2}}} \right)\\
	& =e\sqrt{\log J}\max\left(\twonorma{ \sum_{l\in\compress{\copies}}\tempmat_l\tempmat_l^{*}}^{\frac{1}{2}},\twonorma{ \sum_{l\in\compress{\copies}}\tempmat_l^{*}\tempmat_l} ^{1/2}\right)\\
	& \le e\sqrt{\log J}\max\left((\sum_{l\in\compress{\copies}} \twonorma{\tempmat_l}^{2})^{\frac{1}{2}},(\sum_{l\in\compress{\copies}}\twonorma{\tempmat_l} ^{2})^{\frac{1}{2}}\right)\\
	& =e\sqrt{\log J}\cdot(\sum_{l\in\compress{\copies}}\twonorma{ \tempmat_l}^{2})^{\frac{1}{2}},
\end{align*}
as claimed. We assumed above that $J\ge e^2$ 
to produce the first line, and the second to the last line above uses the triangle inequality and the fact that $\twonorm{ \tempmat \tempmat^{*}}=\twonorm{\tempmat}^{2}$ for any matrix $\tempmat$. We remark that had we stopped at the fifth line above, we would have ended with the stronger original non-commutative Khintchine inequality for the spectral norm. However, the weaker bound given in the last line suffices for our purposes in this paper and completes the proof of Lemma~\ref{lem:weak kintchine ineq}.

\end{proof}

We also list two trivial identities regarding covering numbers, which hold for every set $\genset$, norm $\gennorm{\cdot}$, and $\covres,\dumone>0$:
\begin{align}
	\cover{\genset}{\dumone\gennorm{\cdot}}{\covres}&=\cover{\genset}{\gennorm{\cdot}}{\covres/\dumone}\nonumber\\
	\cover{\dumone\genset}{\gennorm{\cdot}}{\covres}&=\cover{\genset}{\gennorm{\cdot}}{\covres/\dumone}.\label{eq:conv fact about covering}
\end{align}
For reference, we also provide an estimation for two integrals we encounter in the analysis:
\begin{equation}
	\int_{0}^{\dumone}\log\left(1+\frac{\dumtwo}{\dummy}\right)\, d\dummy\lesssim\dumone\log\left(1+\frac{\dumtwo}{\dumone}\right).\label{eq:log integral 1}
\end{equation}
\begin{equation}
	\int_{0}^{\dumone}\sqrt{\log\left(1+\frac{\dumtwo}{\dummy}\right)}\, d\dummy\lesssim\dumone\sqrt{\log\left(1+\frac{\dumtwo}{\dumone}\right)},\label{eq:log integral 2}
\end{equation}
Both inequalities hold when $\dumone\le\dumtwo$.

\section{Proof of Lemma~\ref{lem:coh of rand basis}}
\label{sec:proof of coh of rand basis}

%
Equivalently, $\rbasis$ can be constructed as follows. Let $\{\rbasisc_{\bign}\}$,  $\bign\in\compress{\bigN}$, denote the columns of $\rbasis$. 
	The first column, $\rbasisc_1$, is chosen from the uniform distribution on the unit sphere
	in $\reals^{\bigN}$. 
	For every $\bign\in\compress{\bigN}\backslash\{1\}$, $\rbasisc_{\bign}$ is chosen from the uniform distribution on the unit sphere in the orthogonal complement of the span of the first $\bign-1$ columns.

Let $\rgvec\in\reals^{\bigN}$ denote a standard Gaussian vector, that is a vector whose entries are i.i.d.\ zero-mean Gaussian random variables with unit variance.
Since $\rbasisc_1$ is drawn from the uniform distribution on the unit sphere in $\reals^{\bigN}$, the entries of $\rbasisc_1$ have the same distribution as those in $\rgvec/\twonorm{\rgvec}$. Since the distribution of $\rbasis$ remains unchanged under permutation of its rows, every column of $\rbasis$ has the same (marginal) distribution as $\rgvec/\twonorm{\rgvec}$. This, alongside with the union bound, allows us to write the following for any $\dev>0$:
\begin{align*}
	\Proba{\cohU{\rbasis}>\dev\sqrt{\log{\bigN}}}{} & =\Proba{\max_{\bign\in\compress{\bigN}}  \infnorm{\rbasisc_{\bign}}>\dev\sqrt{\log\bigN}/\sqrt{\bigN}}{}\nonumber\\
	& \le\bigN\cdot\max_{\bign\in\compress{\bigN}}\Proba{ \infnorm{\rbasisc_{\bign}}>\dev\sqrt{\log\bigN}/\sqrt{\bigN}}{}\nonumber
	\\ & =\bigN\cdot\Proba{\max_{\bign\in\compress{\bigN}}\frac{\abs{g(\bign)}}{\twonorm{g}}>\dev\sqrt{\frac{\log\bigN}{\bigN}}}{}\nonumber\\
	& \le \bigN^2\cdot\Proba{\frac{\abs{g(1)}}{\twonorm{g}}>\dev\sqrt{\frac{\log\bigN}{\bigN}}}{}\nonumber\\
	& = \bigN^2\cdot\Proba{\frac{\abs{g(1)}}{\twonorm{g}}>\frac{\dev/2\cdot\sqrt{\log\bigN}}{(1-1/2)\sqrt{\bigN}}}{}\nonumber\\
	& \le \bigN^2\cdot\Proba{\abs{g(1)}>\dev/2\cdot\sqrt{\log\bigN}}{}+\bigN^2\cdot\Proba{\twonorm{g}<(1-1/2)\sqrt{\bigN}}{}\nonumber\\
	& \lesssim \bigN^2e^{-\frac{\Cl{g tail}}{4\sgnorm{\rgvec(1)}^2}\dev^2\log\bigN}+\bigN^2e^{-\Cl{rgvec norm}\min\left(\frac{1}{4\sgnorm{\rgvec(1)}^4},\frac{1}{2\sgnorm{\rgvec(1)}^2}\right)\bigN}\nonumber\\
	& \le \bigN^2e^{-\frac{\Cr{g tail}}{4\sgnorm{\rgvec(1)}^2}\dev^2\log\bigN}+\bigN^2e^{-\frac{\Cr{rgvec norm}}{4\sgnorm{\rgvec(1)}^4}\bigN},
\end{align*}
where we used \eqref{eq:bound on signel gaussian rv} and \eqref{eq:bound on norm of Gaussian vector} to bound the failure probability. The last line holds because $\sgnorm{\rgvec(1)}=\sqrt{2/\pi}$.\footnote{
This is easily verified using the moments of the Gaussian distribution.}
If we take
$\bigN\ge\dev^2\log\bigN$, we obtain that
\begin{equation}
	\Proba{\cohU{\rbasis}>\dev\sqrt{\log{\bigN}}}{} \lesssim \bigN^{2-\frac{\Cr{g tail}}{4\sgnorm{\rgvec(1)}^2}\dev^2}+\bigN^{2-\frac{\Cr{rgvec norm}}{4\sgnorm{\rgvec(1)}^4}\dev^2}.
\end{equation}
We arrive at the advocated result when $\dev\gtrsim1$:
\begin{equation}
	\Proba{\cohU{\rbasis}>\dev\sqrt{\log{\bigN}}}{} \lesssim \bigN^{-\dev}+\bigN^{-\dev}= 2\bigN^{-\dev}.
\end{equation}

\section{Proof of Lemma~\ref{lem:Q of random matrix}\label{sec:Proof-of-Lemma Q for random mtx}}

We use here the construction of $\rbasis$ laid down in the beginning of  \ref{sec:proof of coh of rand basis}.
As pointed out in the proof of Lemma~\ref{lem:coh of rand basis}, the columns of $\rbasis$ are dependent but identically distributed as $\rgvec/\twonorm{\rgvec}$, where $\rgvec$ was defined there. Now, for every $\dev$, we can write
\begin{align}
	\Proba{\QU{\rbasis}\gtrsim 1+\sqrt{\frac{J}{N}}+\dev}{} & =\Proba{\max_{\bign\in\compress{\bigN}}\twonorm{\XvecU{\rbasis}{e_{\bign}}} \gtrsim\frac{1}{\sqrt{N}}+\frac{1+\dev}{\sqrt{J}}}{} \nonumber\\
	& \le\bigN\cdot\max_{\bign\in\compress{\bigN}}\Proba{\twonorm{\XvecU{\rbasis}{e_{\bign}}} \gtrsim\frac{1}{\sqrt{N}}+\frac{1+\dev}{\sqrt{J}}}{} \nonumber\\
	& =\bigN\cdot\Proba{\twonorm{\XvecU{\rbasis}{e_{1}}} \gtrsim\frac{1}{\sqrt{N}}+\frac{1+\dev}{\sqrt{J}}}{}.\label{eq:expand Q of random matrix}
\end{align}
The second line uses the union bound and the last line holds due to the identical distribution of the columns of $\rbasis$. It remains to find an upper bound for the probability in the last line above. 
Recall that $\rbasisc_1$ has the same distribution as $\rgvec/\twonorm{\rgvec}$, and thus $\XvecU{\rbasis}{e_{1}}$ has the same distribution as $\rgmat/\frobnorm{\rgmat}$, where $\rgmat\in\comps^{N\times J}$ is formed by reshaping $\rgvec$. Therefore, $\twonorm{\XvecU{\rbasis}{e_{1}}}$ has the same
distribution as $\left\Vert G\right\Vert _{2}/\left\Vert G\right\Vert _{F}$. For fixed $\dev\le1$, the following convenient inequality holds:\footnote{
	This inequality easily follows from the fact that  $1+\dev\ge\left(1+\dev/3\right)\left(1-\dev/3\right)^{-1}$ holds for every $\dev\le 1$.
}
\begin{equation}
	\frac{1}{\sqrt{N}}+\frac{1+\dev}{\sqrt{J}} \gtrsim \frac{(1+\dev/3)\sqrt{N}+\sqrt{J}}{(1-\dev/3)\sqrt{\bigN}}.
	\label{eq:conv inequality}
\end{equation}
Now, we can write that
\begin{align*}
	& \Proba{\twonorm{\XvecU{\rbasis}{e_{1}}}\gtrsim\frac{1}{\sqrt{N}}+\frac{1+\dev}{\sqrt{J}}}{} \\
	& =\Proba{ \frac{\twonorm{\rgmat}}{\frobnorm{\rgmat}}\gtrsim\frac{1}{\sqrt{N}}+\frac{1+\dev}{\sqrt{J}}}{} \\
	& \le\Proba{ \frac{\twonorm{\rgmat}}{\frobnorm{\rgmat}}\gtrsim\frac{\left(1+\frac{\dev}{3}\right)\sqrt{N}+\sqrt{J}}{\left(1-\frac{\dev}{3}\right)\sqrt{\bigN}}}{} \\
	& \le\Proba{\twonorm{\rgmat}\gtrsim\left(1+\frac{\dev}{3}\right)\sqrt{N}+\sqrt{J}}{} +\Proba{\frobnorm{\rgmat}\lesssim\left(1-\frac{\dev}{3}\right)\sqrt{\bigN}}{} \\
	& \lesssim e^{-\frac{1}{18}\dev^{2}N}+e^{-\Cr{rgvec norm}\min\left(\frac{\dev^2}{9\sgnorm{\rgvec(1)}^4}\,,\,\frac{\dev}{3\sgnorm{\rgvec(1)}^2}\right)\bigN}\\
	& \le e^{-\frac{1}{18}\dev^{2}N}+e^{-\frac{\Cr{rgvec norm}}{9\sgnorm{\rgvec(1)}^4}\dev^2\bigN}\\
	& \le2e^{-\Cl{RBD rand basis}\dev^{2}N},
 \end{align*}
where the second line uses \eqref{eq:conv inequality}. The fourth line uses the inequalities \eqref{eq:bound on norm of Gaussian vector} and \eqref{eq:bound on norm of Gaussian matrix} and the fact that $\frobnorm{\rgmat}=\twonorm{\rgvec}$. The second to last line follows because $\sgnorm{\rgvec(1)}=\sqrt{2/\pi}$ and $\dev\le1$. The above upper bound in combination with \eqref{eq:expand Q of random matrix} leads us to the following conclusion:
\begin{equation*}
	\Proba{\QU{\rbasis}\gtrsim 1+\sqrt{\frac{J}{N}}+\dev}{} \lesssim\bigN e^{-\Cr{RBD rand basis}\dev^{2}N}.
\end{equation*}
We complete the proof of Lemma~\ref{lem:Q of random matrix} by taking $N\ge2\Cr{RBD rand basis}^{-1}\dev^{-2}\log\bigN$.

\section{Proof of Lemma~\ref{lem:orthobasis for Toeplitz}}\label{sec:orthobasis for Toeplitz}

Let $\circbasis':=F_{J}\otimes I_{P}\in\comps^{PJ\times PJ}$, where $\otimes$ stands for the Kronecker product. Here, $F_J$ and $I_P$ are the Fourier and canonical orthobases of size, respectively, $J\times J$ and $P\times P$. We consider the natural partitioning of $\circbasis'$ into $J$ row-submatrices of size $P\times PJ$, denoted by $\circbasis'_{j}$,  $j\in\compress{J}$.
Now, for every $j\in\compress{J}$, cyclically shift the rows of $\circbasis'_{j}$ by $j-1$ times upward and create $\circbasis_{j}\in\comps^{P\times PJ}$. Then form the matrix  $\circbasis\in\comps^{PJ\times PJ}$ by replacing every $\circbasis'_{j}$ with $\circbasis_{j}$. The transformation of $\circbasis'$ to $\circbasis$ is visualized in Figure~\ref{fig:T n Tp}.
\begin{figure}
	\begin{subfigure}[b]{.45\textwidth}
	\begin{center}
	\includegraphics[scale=.4]{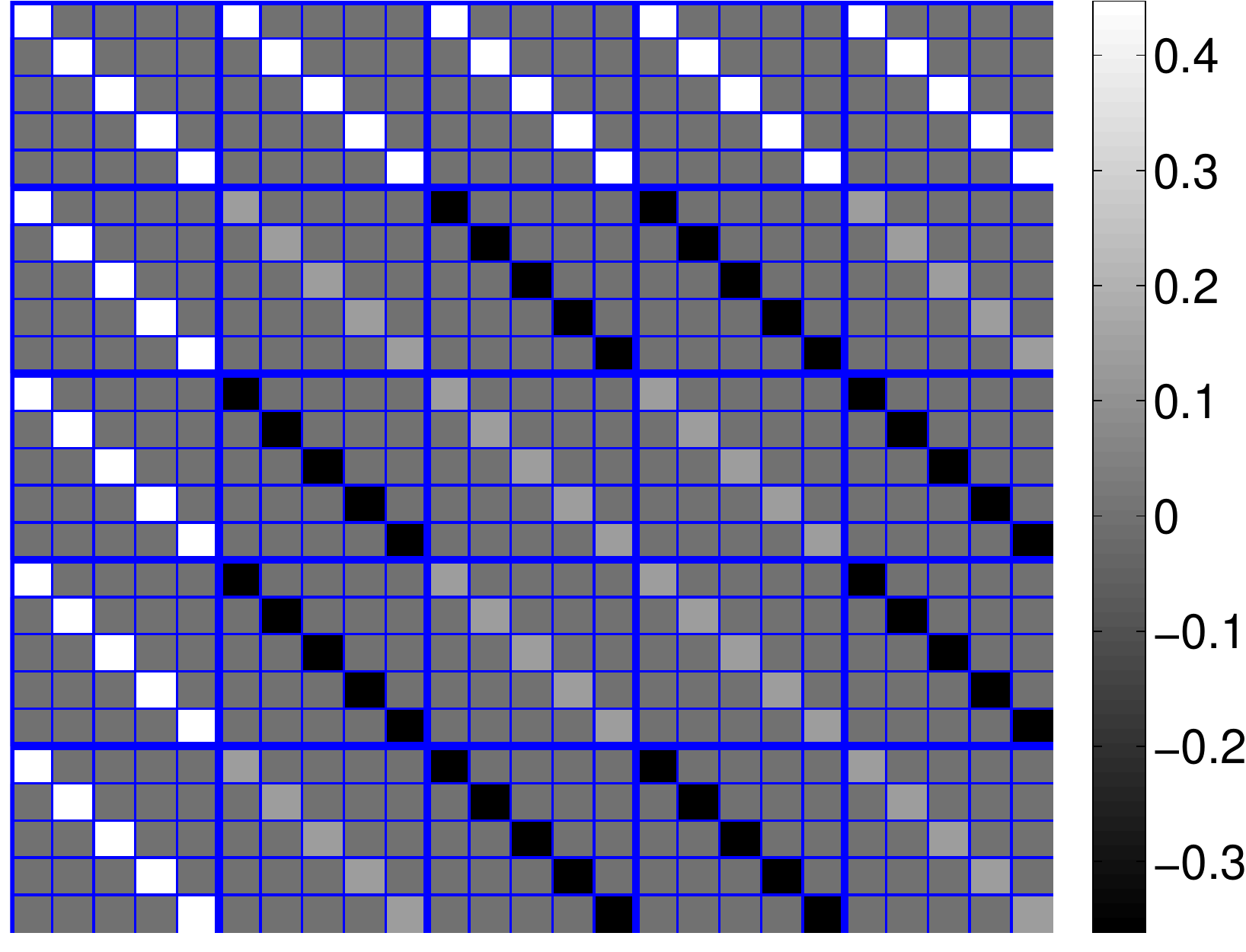}
	\caption{$\re{T'}$}
	\end{center}
	\end{subfigure}
	\begin{subfigure}[b]{0.45\textwidth}
	\begin{center}
	  \includegraphics[scale=.4]{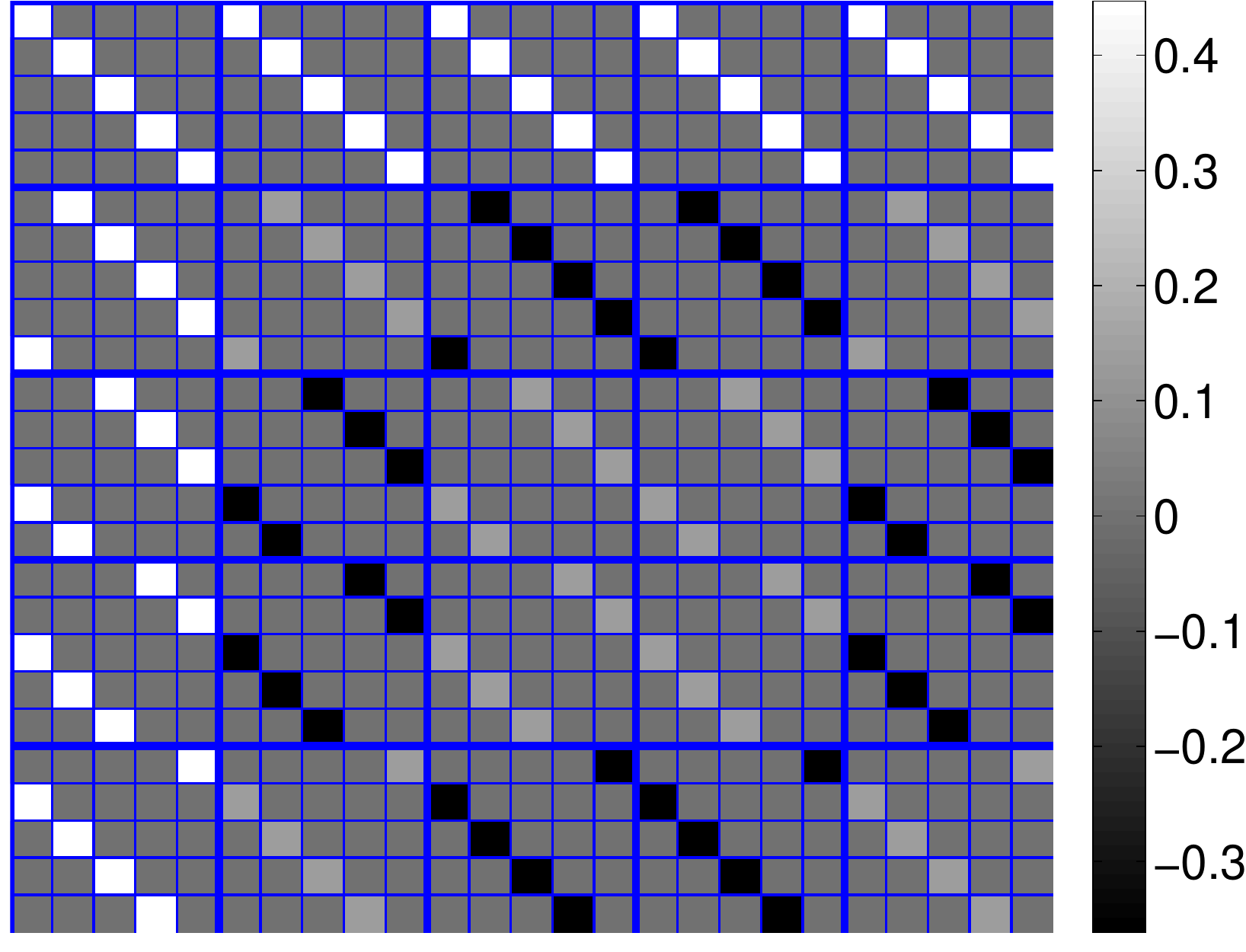}
	  \caption{$\re{T}$}
	\end{center}
	\end{subfigure}	
	\caption{A visual illustration of the transformation from $T'$ to $T$ for $P=J=5$.
	See the explanation in \ref{sec:orthobasis for Toeplitz}.}
	\label{fig:T n Tp}
\end{figure}

Using the properties of the Kronecker product, it is easily verified that $\circbasis'$ is an orthobasis for $\comps^{PJ}$.
Due to its structure, we also observe that the nonzero entries of the $p_1$th and $p_2$th columns of $\circbasis$ do not overlap when $(p_1-p_2)\,\mod\, P\ne0$, and so these columns are orthogonal. Otherwise, when $(p_1-p_2)\,\mod\, P=0$, we need a more subtle argument. Under this condition, the inner product of the $p_1$th and $p_2$th columns of $\circbasis'$ equals the inner product of the $\lceil p_1/P\rceil$th and $\lceil p_2/P\rceil$th columns of $F_J$ and is indeed zero. The inner product of the $p_1$th and $p_2$th columns remains zero under the transformation of $\circbasis'$ to $\circbasis$ since this transformation only amounts to a permutation in the rows of $\circbasis'$.
Therefore, $\circbasis$ is an orthobasis for $\comps^{PJ}$. As for computing $\QU{\circbasis}$, the structure of $\circbasis$ guarantees that, for each $j$, every column  and row of $\XvecU{e_{j}}{\circbasis}$ has only one nonzero entry with the magnitude of $1/\sqrt{J}$. Therefore, $\QU{\circbasis}=1$.

Finally, it can be easily verified that  $\extx/\sqrt{J}=\circbasis\extxtwo$, where $\extxtwo=[x^T,0,0,\cdots,0]^T\in\comps^{PJ}$, and by \eqref{eq:PC_RBD_formulation}, $\circmat x=\circrbd \circbasis\extxtwo$. If $x$ is $S$-sparse, so is $\extxtwo$. This completes the proof of Lemma~\ref{lem:orthobasis for Toeplitz}.

%

\section{Proof of Lemma~\ref{lem: main cover no}}
\label{sec: main cover no RBD}

The arguments used in this section are largely adapted from~\cite{rudelson2008fourier}.
In what follows, we let $\ball{\bigN}{A}$ denote the unit  ball with respect to the norm $\|\cdot\|_A$ in $\comps^{\bigN}$.  Also, $\ball{\bigN}{1}$ and $\ball{\bigN}{2}$, respectively, denote the unit $\lpnorm{1}$-ball and $\lpnorm{2}$-ball in $\comps^{\bigN}$.  For $\supp\subseteq\compress{\bigN}$, we let $\ball{\supp}{A}$
denote the unit $\|\cdot\|_A$ ball in the $\#\supp$-dimensional subspace of $\comps^{\bigN}$  spanned by $\{e_{\bign}\}$, ${\bign\in\supp}$. We define $\ball{\supp}{1}$ and $\ball{\supp}{2}$ similarly.

The first thing to notice is that when $\alpha\in\sparses{S}/\sqrt{S}$, then
$\onenorm{\alpha}\le1$. From the hypothesis of the lemma, we then have that
\begin{equation}
	\|\alpha\|_A\le \frac{\Qall}{\sqrt{\bigM}}.
	\label{eq:bnd on X norm in Omega}
\end{equation}
This implies that for every support $\supp\subset\compress{\bigN}$ with $\#\supp=S$, we have
\begin{equation}
	\frac{\ball{\supp}{2}}{\sqrt{S}}\subseteq \frac{\Qall}{\sqrt{\bigM}}\cdot \ball{\supp}{A}.\label{eq:containment for x ball 1st}
\end{equation}
On the other hand, $\sparses{S}/\sqrt{S}$ can be equivalently represented as
\begin{equation}
	\frac{\sparses{S}}{\sqrt{S}}=\bigcup_{\#\supp=S}\frac{\ball{\supp}{2}}{\sqrt{S}}.\label{eq:alt def of Omega}
\end{equation}
Together, \eqref{eq:containment for x ball 1st} and \eqref{eq:alt def of Omega} imply that
\begin{equation}
	\frac{\sparses{S}}{\sqrt{S}}\subseteq\bigcup_{\#\supp=S}\frac{\Qall}{\sqrt{\bigM}}\cdot \ball{\supp}{A}.\label{eq:containment for Omega 1st}
\end{equation}
We also record that
\begin{equation}
	\frac{\sparses{S}}{\sqrt{S}}\subseteq \frac{\Qall}{\sqrt{\bigM}}\cdot \ball{\bigN}{A},
	\label{eq:containment for Omega}
\end{equation}
which dictates that
\begin{equation}\label{eq:off limit cover no}
	\Cover(\sparses{S}/\sqrt{S},\|\cdot\|_A,\dummy)=1
\end{equation}
if $\dummy\ge\Qall/\sqrt{\bigM}$. This proves the second statement in Lemma~\ref{lem: main cover no}. Otherwise, if $\dummy<\Qall/\sqrt{\bigM}$, we continue with the rest of the argument.
In light of \eqref{eq:containment for Omega 1st}, we have that
\begin{align}\label{eq:cov no of X norm 1 pre raw!}
	\Cover\left(\frac{\sparses{S}}{\sqrt{S}},\|\cdot\|_A,\dummy\right) & \le\Cover\left(\bigcup_{\#\supp=S}\frac{\Qall}{\sqrt{\bigM}}\cdot \ball{\supp}{A},\|\cdot\|_A,\dummy\right)\nonumber \\
	& \le{\bigN \choose S}\cdot\covera{\frac{\Qall}{\sqrt{\bigM}}\cdot\ball{\supp}{A}}{\|\cdot\|_A}{\dummy}\nonumber \\
	&  ={\bigN \choose S}\cdot\covera{\ball{\supp}{A}}{\|\cdot\|_A}{\dummy\Qall^{-1}\sqrt{\bigM}},
\end{align}
where the last line uses the second inequality in \eqref{eq:conv fact about covering}.
An estimate for the covering number in the last line above is available for $\covres>0$, namely\footnote{
This is proved similar to Lemma 5.2 in~\cite{vershynin2010introduction}, after exchanging the Euclidean metric with $\Vert\cdot\Vert_{A}$ and accounting for the complex vector space.}
\begin{equation}
		  \cover{\ball{\supp}{A}}{\|\cdot\|_A}{\covres}\le(1+2/\covres)^{2\#\supp}.\label{eq:cover no of ball with smaller balls}
\end{equation}
Plugging the bound above into \eqref{eq:cov no of X norm 1 pre raw!}, we arrive at
\begin{align}
	\covera{\frac{\sparses{S}}{\sqrt{S}}}{\|\cdot\|_A}{\dummy} \le{\bigN \choose S}\cdot\left(1+\frac{2\Qall}{\dummy\sqrt{\bigM}}\right)^{2S}
	 \le \left(\frac{e\bigN}{S}\right)^{S}\cdot\left(1+\frac{2\Qall}{\dummy\sqrt{\bigM}}\right)^{2S}.\label{eq:cov no of X norm 1 raw}
\end{align}
The last inequality holds because ${n \choose m}\le(en/m)^{m}$ for any pair of integers $m\le n$. When $\bigN\gtrsim1$, \eqref{eq:cov no of X norm 1 raw} implies that
\begin{align}
	\log\covera{\frac{\sparses{S}}{\sqrt{S}}}{\xnorm{\cdot}}{\dummy}
	 & \le S\log\left(\frac{e\bigN}{S}\right)+2S\log\left(1+\frac{2\Qall}{\dummy\sqrt{\bigM}}\right)\nonumber\\
	& \lesssim S\log\bigN+S\log\left(1+\frac{2\Qall}{\dummy\sqrt{\bigM}}\right).
	\label{eq:cov no of X norm 1}
\end{align}
The bound above is less effective for small values of $\dummy$. To seek a second bound on the covering number, we begin from the  containment
\begin{equation}\label{eq:containment of Omega 2nd}
	\frac{\sparses{S}}{\sqrt{S}}\subseteq \ball{\bigN}{1},
\end{equation}
which immediately dictates that
\begin{equation}\label{eq:result of the containment of Omega 2nd}
	\covera{\frac{\sparses{S}}{\sqrt{S}}}{\|\cdot\|_A}{\dummy}\le\covera{\ball{\bigN}{1}}{\|\cdot\|_A}{\dummy}.
\end{equation}
In order to compute the covering number on the right hand side of~\eqref{eq:result of the containment of Omega 2nd}, we use the following result, which is proved in \ref{sec:proof of cover no of real l1 ball}.
\begin{lemma}\label{lem: cover no of real l1 ball}
	Let $B_{1,\bigN}$ denote the $\lpnorm{1}$-ball in $\reals^{\bigN}$, and consider the norm $\|\cdot\|_A$ on $\comps^{\bigN}$, which satisfies the hypothesis of Lemma~\ref{lem: main cover no}. Naturally, $\|\cdot\|_A$ induces a norm on $\reals^{\bigN}\subset\comps^{\bigN}$ (which is represented  with the same notation for convenience). For $\dummy>0$ and $\bigM\gtrsim1$, it holds that
	\[
		\log \left(\covera{B_{1,\bigN}}{\|\cdot\|_A}{\dummy} \right)
		\lesssim \frac{\Qall^{2}}{\dummy^{2}\bigM}\cdot\log^2\bigN.
	\]	
\end{lemma}

Now consider an arbitrary $\temp\in \ball{\bigN}{1}$, and note that $\re{\temp},\,\im{\temp}\in B_{1,\bigN}$. Let $$\acover_1:=\net{B_{1,\bigN}}{\|\cdot\|_A}{\dummy/2}$$ denote a minimal ($\dummy/2$)-cover for $B_{1,\bigN}$ with respect to the metric $\|\cdot\|_A$. Therefore, there exist $p_1,p_2\in \acover_1$ such that $\xnorm{\re{\temp}-p_1},\xnorm{\im{\temp}-p_2}\le \dummy/2$. It follows by the triangle inequality that
\begin{align*}
	\|\temp-(p_1+\iunit p_2)\|_A & =\|(\re{\temp}-p_1)+\iunit (\im{\temp}-p_2)\|_A \\
	& \le \| \re{\temp}-p_1\|_A +  \|\im{\temp}-p_2\|_A\\
	& \le \dummy.
\end{align*}
Therefore, $\{p_1+\iunit p_2 : p_1,\,p_2\in \acover_1\}$ is a cover for $\ball{\bigN}{1}$, and clearly
\[
	\cover{\ball{\bigN}{1}}{\|\cdot\|_A}{\dummy}\le \left(\cover{B_{1,\bigN}}{\|\cdot\|_A}{\dummy/2}\right)^2.
\]
It now follows from \eqref{eq:result of the containment of Omega 2nd}, Lemma~\ref{lem: cover no of real l1 ball}, and the inequality above that
\begin{equation}
	\log\covera{\frac{\sparses{S}}{\sqrt{S}}}{\|\cdot\|_A}{\dummy}
	\lesssim \frac{\Qall^{2}}{\dummy^{2}\bigM}\cdot\log^2\bigN.\label{eq:cov no for X norm 2nd bnd}
\end{equation}
Combining \eqref{eq:cov no of X norm 1} and \eqref{eq:cov no for X norm 2nd bnd}, we finally arrive at
\begin{align*}
	 \log\covera{\frac{\sparses{S}}{\sqrt{S}}}{\|\cdot\|_A}{\dummy}& \lesssim
	 \min\left(S\log\bigN+S\log\left(1+\frac{2\Qall}{\dummy\sqrt{\bigM}}\right),\frac{\Qall^{2}}{\dummy^{2}\bigM}\cdot\log^2\bigN\right),
\end{align*}
which holds when $\bigM\gtrsim1$. This completes the proof of Lemma~\ref{lem: main cover no}.

\section{Proof of Lemma~\ref{lem: cover no of real l1 ball}}
\label{sec:proof of cover no of real l1 ball}

Consider an arbitrary $\temp\in B_{1,\bigN}$. Also consider a random vector $\rv$ that takes a value in $\{0\}\cup\{\sign{\temp(\bign)}\cdot e_{\bign}\}$,
$\bign\in\compress{\bigN}$. It takes $\sign{\temp(\bign)}\cdot e_{\bign}$ with probability of $\abs{\temp(\bign)}$ and zero otherwise.
Clearly, $\Eb{\rv}{}=\temp$. Now, we wish to approximate $\temp$ with the average of $\copies$ independent copies of $\rv$, denoted by $\{\rv_l\}$, $l\in\compress{\copies}$. The expected approximation error, measured in $\|\cdot\|_A$, would be
\begin{equation*}
	\Eb{\left\Vert\temp-\frac{1}{\copies}\sum_{l\in\compress{\copies}}\rv_l\right\Vert_A}{}.
\end{equation*}
Since the argument of the norm is zero-mean, we can use the symmetrization technique by invoking Lemma \ref{lem:sym} from the Toolbox to obtain that
\begin{align}
	\Eb{\left\Vert\temp-\frac{1}{\copies}\sum_{l\in\compress{\copies}}\rv_l\right\Vert_A}{} =
	\frac{1}{\copies}\Eb{\left\Vert\sum_{l\in\compress{\copies}}\rv_l-\Eb{Z_l}{}\right\Vert_A}{}
	\le \frac{2}{\copies}\Eb{\left\Vert\sum_{l\in\compress{\copies}}\rad_l\rv_l\right\Vert_A}{},\label{eq:symm in cool covering}
\end{align}
where $\{\rad_l\}$, $l\in\compress{L}$, is a Rademacher sequence.
According to Lemma~\ref{lem: main cover no}, $\|\cdot\|_A=\|A(\cdot)\|_2$  and $A\left(\cdot\right)$ is a linear map. Therefore
\begin{align*}
	\Eb{\left\Vert\sum_{l\in\compress{\copies}}\rad_{l}\rv_{l}\right\Vert_A}{}
	& =\Eb{\twonorma{\sum_{l\in\compress{\copies}}\rad_lA(\rv_{l})}}{}.
\end{align*}
Also, from the hypothesis of Lemma~\ref{lem: main cover no}, $\rank{A(\rv_l)}\le\bigM$  for every $l$. We now invoke a Khintchine-type inequality for the operator norm (stated in Lemma~\ref{lem:weak kintchine ineq} from the Toolbox), which allows us to continue our argument as
\begin{align}
	\Eb{\twonorma{\sum_{l\in\compress{\copies}}\rad_l A(\rv_{l})}}{}
	& =\Eb{\Eb{\twonorma{\sum_{l\in\compress{\copies}}\rad_l A(\rv_{l})}}{\rad}}{\rv} \nonumber\\
	& \lesssim \sqrt{\log \bigM}\cdot \Eb{\left(\sum_{l\in\compress{\copies}}\twonorm{A(\rv_{l})}^2\right)^{1/2}}{\rv}\nonumber \\
		& \le \sqrt{\log \bigM}\cdot \left(\sum_{l\in\compress{\copies}}\Eb{\twonorm{A(\rv_{l})}^2}{}\right)^{1/2}\nonumber \\	
				& = \sqrt{\log \bigM}\cdot \left(\sum_{l\in\compress{\copies}}\sum_{\bign\in\compress{\bigN}}|\beta_{\bign}|\cdot\twonorm{A(e_{\bign})}^2\right)^{1/2}\nonumber \\	
		& \le \sqrt{\log \bigM}\cdot\left(\copies\cdot\max_{\bign\in\bigN}\|A(e_{\bign})\|^2\right)^{1/2}\nonumber \\
				& = \sqrt{\copies\log \bigM}\cdot \max_{\bign\in\bigN}\|e_{\bign}\|_A\nonumber \\
	& \le \Qall\sqrt{\frac{\copies}{\bigM}\log \bigM},\label{eq:bnd on symm on cool covering}
\end{align}
where, conditioned on $\{\rv_l\}$, $l\in\compress{\copies}$, the Khintchine inequality was used to produce the second line. The third line follows from the Jensen inequality. 
The fifth line uses the assumption that $\beta\in B_{1,\bigN}$.
The last line follows from the hypothesis of Lemma~\ref{lem: main cover no}. Using the above inequality in combination with \eqref{eq:symm in cool covering} yields
\[
	\Eb{\xnorma{\temp-\frac{1}{\copies}\sum_{l\in\compress{\copies}}\rv_{l}}}{}\lesssim\Qall\sqrt{\frac{\log \bigM}{\copies \bigM}}.
\]
To keep the average no larger than $\dummy$, it suffices to take
\[
	\copies\gtrsim\frac{\Qall^{2}}{\dummy^2 \bigM} \cdot\log \bigM.
\]
With this choice of $\copies$, there exists a linear combination of independent copies of $\rv$ that falls within a $\dummy$ distance of $\temp$.
In other words, we have shown that for an arbitrary $\temp \in B_{1,\bigN}$, there exists an average of $\copies$ elements of $\{0\} \cup \{\pm e_{\bign}\}$ that is of distance at most $\dummy$ from $\temp$.
There are $2\bigN+1$ elements in the aforementioned set, and so there are $(2\bigN+1)^{\copies}$
possibilities for the average. We therefore conclude that
\begin{equation*}
	\covera{B_{1,\bigN}}{\|\cdot\|_A}{\dummy}\le\left(2\bigN+1\right)^{\Cl{cover of l1 ball}\log \bigM\cdot\Qall^{2}/\dummy^2 \bigM},
\end{equation*}
or
\begin{align*}
	\log\covera{B_{1,\bigN}}{\|\cdot\|_A}{\dummy}& \lesssim\frac{\Qall^{2}}{\dummy^2 \bigM}\cdot\log \bigM\log\left(2\bigN+1\right)\\
	& \lesssim\frac{\Qall^{2}}{\dummy^2 \bigM}\cdot\log \bigM\log\bigN\\
	& \le \frac{\Qall^{2}}{\dummy^2 \bigM}\cdot\log^2\bigN
\end{align*}
when $\bigN\gtrsim1$. This completes the proof of Lemma \ref{lem: cover no of real l1 ball}.

\end{document}